\documentclass[runningheads]{llncs}
\usepackage{graphicx}
\usepackage{amsfonts}
\usepackage{amssymb}
\usepackage{xcolor}

\usepackage{amsthm}
\RequirePackage{mathtools}
\RequirePackage{subcaption}

\usepackage[colorlinks=true]{hyperref}

\newcommand{\td}{\textnormal{\textsc{td}}}
\newcommand{\tw}{\textnormal{\textsc{tw}}}
\newcommand{\bd}{\textnormal{\textsc{bd}}}
\newcommand{\cc}{\textnormal{\textsc{cc}}}

\newcommand{\tree}{\textnormal{\textsc{tree}}}
\newcommand{\tail}{\textnormal{\textsc{tail}}}
\newcommand{\ed}{\textnormal{\textsc{ed}}}
\newcommand{\fvs}{\textnormal{\textsc{fvs}}}
\newcommand{\FVS}{\textnormal{\textsc{fvs}}}
\newcommand{\prob}[1]{\textnormal{\textsc{#1}}}
\newcommand{\gf}{\mathcal G_F}
\newcommand{\nm}{{\textnormal{\textsc{nm}}}}
\newcommand{\NP}{\mathrm{NP}}
\newcommand{\coNP}{\mathrm{coNP}}
\newcommand{\poly}{\mathrm{poly}}

\ifdefined\DEBUG{}
\def\rem#1{{\marginpar{\raggedright\scriptsize #1}}}

\newcommand{\bmp}[1]{{\color{blue}{#1}}}
\newcommand{\bmpr}[1]{\rem{\textcolor{blue}{\(\bullet \) #1}}}

\else

\newcommand{\bmp}[1]{#1}
\newcommand{\bmpr}[1]{}
\fi

\newtheorem{innercustomthm}{Theorem}
\newenvironment{customthm}[1]
{\renewcommand\theinnercustomthm{#1}\innercustomthm}
{\endinnercustomthm}

\spnewtheorem{innercustomclaim}{Claim}{\itshape}{\rmfamily}

\newtheorem{innercustomlemma}{Lemma}
\newenvironment{customlemma}[1]
{\renewcommand\theinnercustomlemma{#1}\innercustomlemma}
{\endinnercustomlemma}

\newcommand*{\myproofname}{Proof}
\newenvironment{innerproof}[1][\myproofname]{\begin{proof}[#1]
                                                 }{
\end{proof}}

\newcommand{\naturals}{\mathbb{N}}

\newcommand{\bigO}{\mathcal 
O}
\DeclarePairedDelimiter\abs{\lvert}{\rvert}
\DeclarePairedDelimiterX{\set}[1]\lbrace\rbrace{\setaux #1||\endsetaux}
\def\setaux#1|#2|#3\endsetaux{\if\relax\detokenize{#2}\relax #1 \else #1 \;\delimsize\vert\; #2 \fi}

\begin{document}
    \title{Kernelization for Feedback Vertex Set via Elimination Distance to a Forest}
%
%
    \author{David Dekker\inst{1} \and
    Bart M.P.\ Jansen\inst{1}}
    \authorrunning{D.\ J.\ C.\ Dekker and B.\ M.\ P.\ Jansen}
%
    \institute{Eindhoven University of Technology, P.O.\ Box 513, 5600 MB Eindhoven, the Netherlands}
    \maketitle              
    \begin{abstract}
        We study efficient preprocessing for the undirected \prob{Feedback Vertex Set} problem, a fundamental problem in graph theory which asks for a minimum-sized vertex set whose removal yields an acyclic graph.
        More precisely, we aim to determine for which parameterizations this problem admits a polynomial kernel.
        While a characterization is known for the related \prob{Vertex Cover} problem based on the recently introduced notion of bridge-depth, it remained an open problem whether this could be generalized to \prob{Feedback Vertex Set}.
        The answer turns out to be negative; the existence of polynomial kernels for structural parameterizations for \prob{Feedback Vertex Set} is governed by the elimination distance to a forest.
        Under the standard assumption $\NP \not \subseteq \coNP/\poly$, we prove that for any minor-closed graph class $\mathcal{G}$, \prob{Feedback Vertex Set} parameterized by the size of a modulator to $\mathcal{G}$ has a polynomial kernel if and only if $\mathcal{G}$ has bounded elimination distance to a forest.
        This captures and generalizes all \bmp{existing kernels} for structural parameterizations of the \prob{Feedback Vertex Set} problem.

    \end{abstract}

    \section{Introduction}\label{sec:intro}

    For NP-complete problems, a polynomial time algorithm solving any problem instance exactly is unlikely to exist.
However, as one is often interested in solving specific instances, one can try to exploit characteristics of problem instances and develop algorithms that are fast when the input has certain properties.
We therefore associate a parameter with each problem instance.
In our context, a problem instance is a graph for which we ask for the existence of a vertex set of size at most $\ell$ having certain properties.
Such a parameterized instance can be denoted with a triple $(G, \ell, k)$, where we are asking for the existence of a solution of size at most $\ell$ for a graph $G$ with parameter $k$.
We say that an algorithm is \emph{fixed parameter tractable} (FPT) for such a parameterization if it solves any instance $(G, \ell, k)$ of size $n$, as described above, in time bounded by $f(k) n^{\bigO(1)}$ for some computable function $f \colon \naturals \rightarrow \naturals$.

A strongly related field is that of \emph{kernelization}.
This field focuses on reducing a parameterized instance $(G, \ell, k)$ in polynomial time to an equivalent instance $(G', \ell', k')$ whose size is bounded by a computable function of the parameter.
We speak of a polynomial kernel when this function is a polynomial.
It is known that a decidable parameterized problem is fixed parameter tractable if and only if it admits a kernelization~\cite{fptiffkernelization}.
In our quest for determining which parameterizations enable efficient algorithms, it is therefore interesting to determine those that allow a polynomial kernel.

This paper focuses on polynomial kernels for the undirected \prob{Feedback Vertex Set} problem, which is an NP-complete problem in graph theory as originally identified by Karp~\cite{karp}.
For an undirected graph $G$, a vertex set $X \subseteq V(G)$ is a \emph{feedback vertex set} if the graph is acyclic after removal of $X$.
We call such a vertex set whose removal yields a graph in some graph class $\mathcal G$ a $\mathcal G$-modulator and define the deletion distance to $\mathcal G$ as its minimum size.
The \prob{Feedback Vertex Set} problem now asks for the minimum size of such a feedback vertex set, or equivalently, the deletion distance to a forest.
For a graph $G$, we let $\fvs(G)$ (the \emph{feedback vertex number} of $G$) denote that minimum size.
Our main question is for which parameterizations the \prob{Feedback Vertex Set} problem admits a polynomial kernel.

Before exploring the \prob{Feedback Vertex Set} problem further, we should mention the related \prob{Vertex Cover} problem.
It asks for a minimum set of vertices hitting all edges in a graph.
While a kernel in the solution size with a linear number of vertices can be obtained using various techniques~\cite{vclinearc1,vclinearc2,vclinearc3,vclineart1,vclineart2}, a polynomial kernel in a structurally smaller parameter was only discovered in 2011, when Jansen and Bodlaender developed a polynomial kernel in the feedback vertex number of a graph~\cite{vertexcoverfvs}.
From there, many polynomial kernels for \prob{Vertex Cover} were discovered in modulators to even larger graph classes~\cite{vcsmalldeg,vctd1,vctd2,vcdpseudoforest}.
In 2020, Bougeret, Jansen and Sau proved the following characterization under common hardness assumptions: \prob{Vertex Cover} admits a polynomial kernel in the size of a modulator to a minor-closed graph family $\mathcal G$ if and only if $\mathcal G$ has bounded \emph{bridge-depth}~\cite{bridgedepth}.
With this result, they generalized all existing work on kernels in the size of modulators to minor-closed graph families, and they proved that their results cannot be improved further under common hardness assumptions.

These results now bring us back to the \prob{Feedback Vertex Set} problem.
For this problem, finding a polynomial kernel in the solution size had been an open problem for a long time.
The first polynomial kernel with size bound $\bigO(k^{11})$ was obtained in 2006 and it was subsequently improved to a quadratic kernel~\cite{fvspolynomialkernel,fvsqubickernel,fvsquadratickernel}.
After the improvements for \prob{Vertex Cover}, \bmp{researchers} also tried to develop polynomial kernels in smaller parameters for \prob{Feedback Vertex Set}~\cite{jansen2014,fvsparam2,jansenpieterse2018}.
It remained an open problem whether these results could be generalized further or whether there exists some parameter that characterizes \prob{Feedback Vertex Set} similarly to how bridge-depth characterizes \prob{Vertex Cover}.
In particular, Bougeret, Jansen and Sau suggested in their paper on \prob{Vertex Cover} that the deletion distance to constant bridge-depth might also be an interesting parameter to consider for problems such as \prob{Feedback Vertex Set}.
We therefore aim to answer the question for which graph families $\mathcal G$ the \prob{Feedback Vertex Set} problem admits a polynomial kernel when parameterized by the size of a $\mathcal G$-modulator.

\paragraph{Our results}
To our initial surprise, the results for \prob{Vertex Cover} cannot be generalized to \prob{Feedback Vertex Set}.
It turns out that a minor-closed graph family $\mathcal G$ must have bounded \emph{elimination distance to a forest} \bmp{(Definition~\ref{def:eldist})}, in order to allow a polynomial kernel in a $\mathcal G$-modulator.
This concept was introduced by Bulian and Dawar~\cite{eldist} and is another generalization of the more common parameter treedepth~\cite{treedepth}.
Our result is described in Theorem~\ref{thm:main}.

\begin{theorem}\label{thm:main}
Let $\mathcal G$ be a minor-closed graph family.
Then \prob{Feedback Vertex Set} admits a polynomial kernel in the size of a $\mathcal G$-modulator if and only if $\mathcal G$ has bounded elimination distance to a forest, assuming that $\NP \not \subseteq \coNP / \poly$.
\end{theorem}

The minor-closed and hardness assumptions are only needed for the lower bound.
To the best of our knowledge, our kernel generalizes all known polynomial kernels for the \prob{Feedback Vertex Set} problem.
Both the kernel and its correctness proof follow the structure of the kernel for $\mathcal F$-\prob{Minor Free Deletion} in the deletion distance to a graph of constant treedepth by Jansen and Pieterse~\cite{jansenpieterse2018}.
The correctness proof of their kernel crucially relies on their Lemma 3 \bmp{whose technical proof spans thirty pages. We require a variation of this lemma. On the one hand, our variation is more involved since it deals with elimination distance to a forest rather than treedepth; on the other hand, it is simpler since it concerns only \prob{Feedback Vertex Set} rather than $\mathcal F$-\prob{Minor Free Deletion}. As a result of this simplification, we can formulate the lemma in terms of reducing some vertex sets $S$ and $T$. As one of our main contributions, we prove this lemma using a strategy that differs significantly from the one followed in earlier work~\cite{jansenpieterse2018}.}

\begin{lemma}\label{lemma3}
Let $G$ be a connected graph with disjoint vertex sets $S,T \subseteq V(G)$.
Suppose that any minimum feedback vertex set $X$ of $G$ misses some vertex from $S$ or leaves two vertices from $T$ connected in $G-X$.
Then there exist sets $S^* \subseteq S$ and $T^* \subseteq T$ whose sizes only depend on the elimination distance to a forest of $G$, such that any minimum feedback vertex set $X$ of $G$ misses some vertex from $S^*$ or leaves two vertices from $T^*$ connected in $G-X$.
\end{lemma}

\bmp{Once Lemma~\ref{lemma3} is proven, the kernelization upper bounds follow similarly to earlier work~\cite{jansenpieterse2018}.} As for the lower bound in Theorem~\ref{thm:main}, we are also able to generalize our proof for other $\mathcal F$-\prob{Minor Free Deletion} problems as described in Theorem~\ref{thm:main-lower-bound}.

\begin{theorem}\label{thm:main-lower-bound}
    Let $\mathcal G$ be a minor-closed family of graphs and let $\mathcal F$ be a finite set of biconnected planar graphs on at least three vertices.
    If $\mathcal G$ has unbounded elimination distance to an $\mathcal F$-minor free graph, then $\mathcal F$-\prob{Minor Free Deletion} does not admit a polynomial kernel in the size of a $\mathcal G$-modulator, unless $\NP \subseteq \coNP / \poly$.
\end{theorem}

\paragraph{Related work}
FPT algorithms  \bmp{for \textsc{Feedback Vertex Set}} parameterized by the solution size are known with running times $\bigO^*(3.490^k)$ (deterministic) and $\bigO^*(2.7^k)$ (randomized)~\cite{3592,3460,27}.
The parameterization by treewidth has \bmp{led} to deterministic $\bigO^*(11.36^k)$ and randomized $\bigO^*(3^k)$ algorithms~\cite{cutandcount,treewidthdeterm}.
Other FPT algorithms include parameterizations by deletion distance to a chordal graph, by cliquewidth and by the number of vertices of degree at least $4$~\cite{jansen2014,improvedcvd,cliquewidth,fvsparam2}.
Recently, Donkers and Jansen proposed a parameterization based on the complexity of the solution structure~\cite{antlers}.
They introduced the notion of \emph{antlers}, similar to the notion of crown decompositions for \prob{Vertex Cover}.

\paragraph{Organization}
Section~\ref{sec:prelim} introduces all relevant terminology.
Section~\ref{sec:elim-dist-proof} presents our kernel and thereby proves the `if' direction of Theorem~\ref{thm:main}.
Then Section~\ref{sec:lower-bound} contains the proof of Theorem~\ref{thm:main-lower-bound}, thereby also proving the `only if' direction of Theorem~\ref{thm:main}.
Lastly, Section~\ref{sec:discussion} contains our conclusions and discusses future work.

    \section{Preliminaries}\label{sec:prelim}

    For a positive integer $n$, we use the shorthand $[n]$ for the set of all natural numbers $i$ with $1 \le i \le n$.
All graphs we consider are finite, undirected and simple.
When $G$ is a graph, we let $V(G)$ denote the vertex set of $G$ and $E(G)$ the edge set.
For $S \subseteq V(G)$, the graph $G-S$ is the graph where all vertices in $S$ and all incident edges are removed, and the graph $G[S]$ is the subgraph of $G$ induced by the vertices in $S$.
When an edge exists between two vertices in $G$, we say that these vertices are \emph{adjacent}.
The \emph{neighbors} of $v$ in $G$, denoted with $N_G(v)$, are the vertices adjacent to a vertex $v \in V(G)$ in $G$.
For $S \subseteq V(G)$, we say that $v \in V(G-S)$ is adjacent to $S$ if there exists some edge between $v$ and a vertex in $S$.
The set $N_G(S)$ contains all vertices $v \in V(G - S)$ for which this holds.
We will sometimes slightly abuse notation and speak of a vertex being adjacent to some subgraph, rather than to the vertices in that subgraph.
We say that two vertices are \emph{connected} in $G$ when they are in the same connected component.
The set $\cc(G)$ denotes the set of connected components (or shortly components) of $G$.
For sets $S, T \subseteq (V)$, we say that $S$ \emph{separates} $T$ if each component of $G-S$ contains at most one vertex from $T$.
Notice that we do not require $S$ and $T$ to be disjoint.
Such a set $S$ is a \emph{vertex multiway cut} of $T$ in $G$.
We use the notation $\bigO_{\eta}(f(n))$ to describe the functions in $\eta$ and $n$ which can be bounded by $g(\eta) \cdot f(n)$ for some computable function $g$.

A concept that will be used extensively is the concept of graph minors.
This uses the notion of \emph{edge contraction}.
When $uv$ is an edge in a graph $G$, contracting this edge replaces vertices $u$ and $v$ by a new vertex whose set of neighbors is $N_G(\set{u, v})$.
Now $H$ is a \emph{minor} of $G$ if $H$ can be obtained from $G$ by removing vertices, removing edges and contracting edges.
Alternatively, one can define $H$ to be a minor of $G$ if there exists a \emph{minor model} $\phi \colon V(H) \to 2^{V(G)}$, such that for any $v \in V(H)$ the graph $G[\phi(u)]$ is connected, for any distinct $u, v \in V(H)$ we have $\phi(u) \cap \phi(v) = \emptyset$, and for any edge $uv$ in $H$ there exists an edge between a vertex in $\phi(u)$ and a vertex in $\phi(v)$ in $G$.

A graph $G$ has an $\mathcal H$-\emph{minor} for a set of graphs $\mathcal H$ if $G$ contains some graph $H \in \mathcal H$ as a minor.
For a minor model $\phi$ of $H$ in $G$ and a set $S \subseteq V(H)$, we use the shorthand notation $\phi(S) \coloneqq \bigcup_{v \in S} \phi(v)$.
We say that a minor model $\phi$ of $H$ in $G$ is \emph{minimal}, if there does not exist a minor model $\phi'$ of $H$ in $G$ with $\phi'(V(H)) \subsetneq \phi(V(H))$.
A minor model $\phi$ of $H$ in $G$ and a minor model $\phi'$ of $H'$ in $G$ \emph{intersect} if $\phi(V(H)) \cap \phi'(V(H')) \neq \emptyset$.

\subsection{Elimination distance and treewidth}\label{subsec:elim-distances}
Our work relies crucially on the concept of elimination distance as introduced by Bulian and Dawar~\cite{eldist}.
\begin{definition}[Elimination distance]
    \label{def:eldist}
    Let $G$ be a graph and $\mathcal G$ a graph family.
    Then the elimination distance of $G$ to $\mathcal G$ is
    \[
        \ed_{\mathcal G}(G) = \begin{cases}
                                  0 &\text{if $G \in \mathcal G$,} \\
                                  \max_{G' \in \cc(G)} \ed_{\mathcal G}(G') &\text{if $\abs{\cc(G)} > 1$,} \\
                                  \min_{v \in V(G)} \ed_{\mathcal G}(G - \set v) + 1 &\text{otherwise.}\\
        \end{cases}
    \]
\end{definition}
We only consider graph families $\mathcal G$ that are minor-closed.
We use the shorthand $\gf$ for the graph family containing precisely all forests.
We say that a graph family $\mathcal G$ has bounded elimination distance to some graph class $\mathcal H$ if all graphs in $\mathcal G$ have bounded elimination distance to $\mathcal H$.
The elimination distance of a graph $G$ to the empty graph is called the \emph{treedepth} of $G$ and is denoted with $\td(G)$.
More intuitively, the elimination distance to a graph class $\mathcal G$ can be interpreted as the minimum number of `elimination iterations' that are necessary to obtain a graph where every connected component is in $\mathcal G$.
In such an iteration, one is allowed to remove one vertex from each connected component.
This interpretation leads to the notion of an elimination forest.

\begin{definition}[$\mathcal G$-elimination forest]
    \label{def:elim-dist-decomp}
    Let $G$ be a graph and $\mathcal G$ a graph family.
    A $\mathcal G$-elimination forest of $G$ is a tuple $\left(F, (B_u)_{u \in V(F)}\right)$ where $F$ is a rooted forest and where each vertex $v \in V(F)$ has a bag $B_v \subseteq V(G)$ such that:
    \begin{itemize}
        \item The bags define a partition of $V(G)$, i.e.\ for any vertex $v \in V(G)$ there is a unique vertex $u \in V(F)$ with $v \in B_u$.
        \item For any non-leaf $u$ of $F$, its bag $B_u$ contains precisely one vertex.
        \item For any leaf $u$ of $F$, the graph $G[B_u]$ is connected and $G[B_u] \in \mathcal G$.
        \item For any edge $uv$ in $G$, let $s, t \in V(F)$ be the vertices such that $u \in B_s$ and $v \in B_t$.
        Then $s$ is an ancestor of $t$, or $t$ is an ancestor of $s$ in $F$.
    \end{itemize}
\end{definition}

We define the height of an elimination forest $F$ to be the maximum number of edges on a path from the root to a leaf in $F$.
One can prove with induction that this height is equal to the elimination distance we defined earlier.

We will use these elimination forests extensively for our kernel and therefore introduce some shorthand notation.
Let $(F, (B_u)_{u \in V(F)})$ be a $\mathcal G$-elimination forest.
Let $v$ be a vertex in $F$.
The \emph{tail} of $v$, denoted with $\tail(v)$, is defined as the union of $B_u$ over all proper ancestors $u$ of $v$.
The closed tail $\tail[v]$ also includes $B_v$.
Similarly, $\tree(v)$ denotes the union of $B_u$ over all proper descendants $u$ of $v$ and $\tree[v]$ also includes $B_v$.
The subgraph of $G$ induced by all vertices in $\tree[v]$ is denoted with $G_v$.
We will sometimes slightly abuse notation and use $G_v$ as a vertex set.
We use the shorthand $G_v^+$ to describe the induced subgraph on the vertices in $\tree[v] \cup \tail[v]$.

We will also introduce the notion of \emph{bridge-depth} as introduced by Bougeret, Jansen and Sau~\cite{bridgedepth}.
A \emph{bridge} in a graph $G$ is an edge whose removal increases the number of connected components of $G$.
The concept of bridge-depth now allows us to delete a set of vertices $S$ as long as $G[S]$ is connected and each edge in $G[S]$ is a bridge in $G$.
Such a structure $G[S]$ is called a \emph{tree of bridges}. 
Observe that a single vertex is always a tree of bridges.
\begin{definition}[Bridge-depth]
    \label{def:bd}
    Let $G$ be a graph.
    The bridge-depth of $G$ is defined as
    \[
        \bd(G) = \begin{cases}
                     0 &\text{if $G$ is the empty graph,} \\
                     \displaystyle \max_{G' \in \cc(G)} \bd(G') &\text{if $\abs{\cc(G)} > 1$,} \\
                     \displaystyle \min_{\substack{S \subseteq V(G): \\ \text{$G[S]$ is a tree of bridges}}} \bd(G - S) + 1 &\text{otherwise.}\\
        \end{cases}
    \]
\end{definition}
Cf.~\cite{bridgedepth} for equivalent definitions.
Lastly, we sometimes use the more common concept of \emph{treewidth} which can be defined using a \emph{tree decomposition}.

\begin{definition}[Tree decomposition]
    Let $G$ be a graph.
    A tree decomposition of $G$ is a tuple $(R, (B_u)_{u \in V(R)})$ where $R$ is a rooted tree and where each node $u \in V(R)$ has a bag $B_u \subseteq V(G)$ with the following properties:
    \begin{itemize}
        \item For every vertex $v \in V(G)$, there exists a node $u \in V(R)$ with $v \in B_u$.
        \item For every edge $vw \in E(G)$, there exists a node $u \in V(R)$ with $\set{v, w} \subseteq B_u$.
        \item For every vertex $u \in V(G)$, the subgraph of $R$ induced by all vertices whose bag contains $u$ is a tree.
    \end{itemize}
\end{definition}

For $u \in V(R)$ in the setting above, we let $G_u$ be the subgraph of $G$ induced by the union of $B_v$ over all descendants $v$ of $u$ in the rooted tree $R$, i.e.
\[
    G_u = G\left[ \bigcup_{v \in V(R_u)} B_v \right].
\]
The \emph{width} of a tree decomposition is defined as the size of its largest bag minus 1.
The \emph{treewidth} of a graph $G$, denoted with $\tw(G)$, can now be defined as the minimum width over all tree decompositions of $G$.
We mention some useful properties of these concepts in Proposition~\ref{prop:ed-properties}.

\begin{proposition}
    \label{prop:ed-properties}
    Let $G$ be a graph with $\mathcal G_F$-elimination forest $\left(F, (B_u)_{u \in V(F)}\right)$ and let $\eta$ be an integer such that $\ed_{\mathcal G_F}(G) \le \eta$.
    Let $X$ be a minimum feedback vertex set in $G$ and let $v$ be a leaf in $F$.
    Then the following claims hold. 
    \begin{enumerate}
        \item\label{prop:tw} $\tw(G) \le \bd(G) \le \ed_{\gf}(G) + 1$.
        \item\label{prop:X-cap-Bv} $X$ contains at most $\eta$ vertices from $B_v$.
        \item\label{prop:path-via-tail} If there exists a path in $G$ from a vertex in $B_v$ to a vertex outside $B_v$, then this path contains a vertex in $\tail(v)$.
        \item\label{prop:X-cap-Gc} Let $u \in V(F)$, then $u$ has at most $\eta$ children $c$ where $X \cap G_c$ is not a minimum feedback vertex set in $G_c$.
    \end{enumerate}
\end{proposition}
\begin{proof}
    \textbf{Proof of~\ref{prop:tw}.}
    Observe that a single vertex is also a tree of bridges.
    One can therefore derive that $\bd(G) \le \ed_{\gf}(G) + 1$.
    There is a difference of one, since a forest has elimination distance 0 to a forest while its bridgedepth is 1.
    Besides, it is known that $\tw(G) \le \bd(G)$ for any graph $G$~\cite[Proposition 3.2]{bridgedepth}, thereby concluding the proof.
    
    \textbf{Proof of~\ref{prop:X-cap-Bv}.}
    Suppose towards a contradiction that $X$ is a minimum feedback vertex set in $G$, but $\abs{B_v \cap X} > \eta$.
    We claim that $X' \coloneqq (X \setminus B_v) \cup \tail(v)$ is a feedback vertex set in $G$ of smaller size.
    Since $\abs{\tail(v)} \le \eta$, it follows that $\abs{X'} < \abs X$.
    Furthermore, suppose towards a contradiction that $X'$ is not a feedback vertex set, then it contains some cycle disjoint from $\tail(v)$.
    However, then this cycle must also be disjoint from $B_v$, as $G[B_v]$ is a forest and $N(B_v) \subseteq \tail(v)$ which are all in $X'$.
    Therefore, this cycle must also be contained in $G-X$ which is a contradiction with $X$ being a feedback vertex set.
    We conclude that any minimum feedback vertex set contains at most $\eta$ vertices from $B_v$.
    
    \textbf{Proof of~\ref{prop:path-via-tail}.}
    If there exists path from a vertex in $B_v$ to a vertex outside $B_v$, then there also exists an edge between a vertex in $B_v$ to a vertex outside $B_v$.
    By definition of an elimination forest and since $v$ is a leaf, this edge has an endpoint in the bag of a proper ancestor of $v$, so it is in $\tail(v)$. 
    
    \textbf{Proof of~\ref{prop:X-cap-Gc}.}
    Let $C$ denote the set of children of $u$ and let $C' \subseteq C$ be those children $w$ where $X \cap G_w$ is not a minimum feedback vertex set in $G_w$.
    Assume towards a contradiction that $\abs{C'} > \eta$, then we are going to construct a feedback vertex set with strictly fewer vertices.
    For all children $w \in C'$, let $X'_w$ be a minimum feedback vertex set in $G_w$.
    We now claim that
    \[
        X' \coloneqq \left( X \setminus \bigcup_{w \in C'} (X \cap G_w) \right) \cup \tail[u] \cup \bigcup_{w \in C'} X'_w
    \]
    is a feedback vertex set in $G$ with less vertices than $X$.
    Since $\abs{X \cap G_w} > \abs{X'_w}$ for each $w \in C'$, $\abs{C'} > \eta$ and $\abs{\tail[u]} \le \eta$, it follows that $\abs{X'} < \abs X - \eta + \eta = \abs X$.

    Observe that any cycle that does not contain vertices in $G_w$ for some child $w \in C'$ must be broken by $X'$, as no vertex from this cycle was removed from $X$ while $X$ is a feedback vertex set.
    Similarly, any cycle that only contains vertices in $G_w$ for some $w \in C'$ must be broken by $Y$, because $X'_w$ is a feedback vertex set in $G_w$.
    Any remaining cycle in $G-X'$ therefore must contain a vertex in $G_w$ for some $w \in C'$ and a vertex outside this subgraph.
    By Proposition~\ref{prop:ed-properties}.\ref{prop:path-via-tail}, any path in this component between a vertex in $G_w$ and a vertex outside $G_w$ must contain a vertex from $\tail(w) = \tail[u]$.
    This implies directly that also any cycle of the last category must be broken by $X'$.
    We can conclude that $X'$ is a feedback vertex set with strictly fewer vertices than $X$, thereby contradicting that $X$ is a minimum feedback vertex set.
\end{proof}

    \section{Kernelization upper bounds}\label{sec:elim-dist-proof}

    We present a kernelization algorithm for Feedback Vertex Set in the deletion distance to a graph of constant elimination distance to a forest.
The structure follows the kernelization algorithm for the generalized $\mathcal F$-Minor Free Deletion problem parameterized by deletion distance to a graph of constant treedepth, as presented by Jansen and Pieterse~\cite[Lemma 3]{jansenpieterse2018}.
The proof of their algorithm requires a lemma on labeled minors.
Since we restrict ourselves to the feedback vertex set problem, our reformulation in Lemma~\ref{lemma3} in the introduction suffices for our purposes.
We repeat and prove this formulation in Section~\ref{subsec:main-lemma-proof} and describe the kernelization algorithm in Section~\ref{subsec:kernel}.
We first need to address some preliminaries.

\subsection{Additional preliminaries for Section~\ref{sec:elim-dist-proof}}\label{app:prelim-elim-dist-proof}
Let $G$ be a graph and $X \subseteq V(G)$.
For our kernel, it will be important to bound the number of relevant components of $G - X$ when $X$ is a modulator to a graph with constant elimination distance to a forest.
This will be done by characterizing how a component $C \in \cc(G-X)$ is connected to $X$.
To this end, we define the following notation.
\begin{definition}[$\mathcal S_G$ and $\mathcal T_G$]
    \label{def:st}
    Let $G$ be a graph with $X \subseteq V(G)$ and $C \in \cc(G-X)$.
    Then
    \begin{gather*}
        \mathcal S_G(C, X) \coloneqq \set{v \in V(C) | \abs{N_G(v) \cap X} \ge 2}, \\
        \mathcal T_G(C, X) \coloneqq \set{v \in V(C) | \abs{N_G(v) \cap X} = 1}.
    \end{gather*}
\end{definition}

These sets characterize how $C$ is connected to the remainder of $G$.
The vertices in $\mathcal S_G(C, X)$ are those vertices in $C$ with at least two neighbors in $X$, while the vertices in $\mathcal T_G(C,X)$ are those with exactly one neighbor.
We will often use these sets in the context of a feedback vertex set missing vertices from $\mathcal S_G(C, X)$ or leaving a pair in $\mathcal T_G(C,X)$ connected.
We therefore use the notation with $\mathcal S$ and $\mathcal T$ in analogy with sets $S$ and $T$ in the formulation of Lemma~\ref{lemma3}.
The following proposition presents such a relation between these sets and feedback vertex sets.
Recall that we say that a vertex set $X$ separates a vertex set $T$ in a graph $G$ if each connected component of $G-X$ contains at most one vertex from $T$.

\begin{proposition}
    \label{prop:samen-fvs}
    Let $G$ be a graph, $X \subseteq V(G)$ and $C^* \in \cc(G-X)$.
    Define $\widehat G = G - V(C^*)$.
    Let $\widehat Y$ be a feedback vertex set in $\widehat G$ and let $Y^*$ be a feedback vertex set in $C^*$.
    Define $X_0 = X \setminus \widehat Y$.
    If $Y^*$ contains $\mathcal S_G(C^*, X_0)$ and separates $\mathcal T_G(C^*, X_0)$, then $\widehat Y \cup Y^*$ is a feedback vertex set in $G$.
\end{proposition}

One should observe that this statement is slightly stronger than requiring that $Y^*$ contains $\mathcal S_G(C^*, X)$ and separates $\mathcal T_G(C^*, X)$.
We will rely on our weakened assumption when applying Proposition~\ref{prop:samen-fvs}.

\begin{proof}[Proof of Proposition~\ref{prop:samen-fvs}]
    Assume towards a contradiction that $\widehat Y \cup Y^*$ is not a feedback vertex set in $G$.
    Then observe that any cycle must contain vertices both from $C^*$ and from $\widehat G$, since any other cycle is broken trivially by either $\widehat Y$ or $Y^*$.
    This implies that there exists a path of at least two vertices in $C^* - Y^*$ whose endpoints have a neighbor in $X_0$: the path cannot contain only one vertex, since such a vertex would have two neighbors in $X$, while $\mathcal S_G(C^*, X_0)$ is included in $Y^*$.
    Therefore, the path has length at least one and both endpoints are in $\mathcal T_G(C^*, X_0)$.
    However, then a pair of vertices in $\mathcal T_G(C^*, X_0)$ is connected in $C^* - Y^*$, reaching a contradiction with the assumption that $Y^*$ separates $\mathcal T_G(C^*, X_0)$.
    It follows that $G - (Y^* \cup \widehat Y)$ is acyclic and that $Y^* \cup \widehat Y$ is a feedback vertex set in $G$.
\end{proof}

The following proposition explains a property of acyclic graphs which is useful in the context of feedback vertex sets.

\begin{proposition}
    \label{prop:niet-zo-veel-paths}
    Let $G$ be an acyclic graph with $X \subseteq V(G)$ and let $\mathcal C = \cc(G-X)$.
    Let $X^* \subseteq X$.
    There are at most $\abs{X^*} - 1$ components $C \in \mathcal C$ for which there exist distinct $u, v \in X^*$ such that some $u$--$v$ path contains a vertex in $C$.
\end{proposition}
\begin{proof}
    Assume towards a contradiction that there exist $\abs{X^*}$ of these components.
    For such a component $C$, let $P_C$ be a minimal path that has its endpoints in $X^*$ while it also contains a vertex in $C$.
    Observe that such a path contains at least three vertices and notice that the minimality of the path implies that all vertices other than the endpoints are in $C$.
    If two of these paths have the same endpoints, we directly obtain a cycle in $G$ consisting of these two paths.
    Otherwise, we create an auxiliary graph $H$ on vertex set $X^*$.
    We add an edge between $z_1, z_2 \in V(X^*)$ if for some component $C$ the path $P_C$ has endpoints $z_1$ and $z_2$.
    Since we add $\abs{X^*}$ unique edges to a graph of at most $\abs{X^*}$ vertices, $H$ contains a cycle.
    Now consider the cyclic walk in $G$ traversing all paths corresponding to the edges in the cycle in $H$.
    We claim that each vertex occurs only once in this walk.
    Observe that no vertex in $X^*$ occurs twice, as only the endpoints of each path corresponding to a child are in $Z$.
    Furthermore, a vertex outside $X^*$ cannot occur on multiple paths since those paths contain vertices from different components.
    Therefore all vertices are unique in the walk, so we obtain a cycle in $G$ which contradicts it being acyclic.
    We conclude that there exist at most $\abs{X^*} - 1$ of these components.
\end{proof}

A relation between Proposition~\ref{prop:niet-zo-veel-paths}, feedback vertex sets and the earlier defined functions $\mathcal S$ and $\mathcal T$ is described in Proposition~\ref{prop:max-gamma-components}.

\begin{proposition}
    \label{prop:max-gamma-components}
    Let $G$ be a graph, $X \subseteq V(G)$ and $\mathcal C = \cc(G-X)$.
    Let $Y$ be a feedback vertex set in $G$ and let $X^* \subseteq X \setminus Y$.
    Then there are at most $\abs{X^*} - 1$ components $C \in \mathcal C$ for which $Y \cap V(C)$ misses a vertex in $\mathcal S_G(C, X^*)$ or leaves a pair of vertices in $\mathcal T_G(C, X^*)$ connected in $C - Y$.
\end{proposition}
\begin{proof}
    Let $\mathcal C'$ denote the set of all components $C \in \mathcal C$ for which $Y \cap V(C)$ misses a vertex in $\mathcal S_G(C, X^*)$ or leaves a pair of vertices in $\mathcal T_G(C, X^*)$ connected in $C - Y$.
    Observe that for none of these components $C$, the graph $C - Y$ contains a path of length one or more where both endpoints are adjacent to the same vertex $u \in X^*$ in $G - Y$, as this contradicts $Y$ being a feedback vertex set.
    Therefore for each $C \in \mathcal C'$, there exists a path in $G-Y$ containing a vertex in $C$ such that both endpoints are distinct vertices from $X^*$.
    By applying Proposition~\ref{prop:niet-zo-veel-paths}, we now obtain the required bound on $\mathcal C'$.
\end{proof}

We will also use the following lemma on the time complexity of determining minimum feedback vertex sets in a graph of bounded elimination distance to a forest.
We should note that both results can also be derived from Courcelle's theorem as explained by Jansen and Pieterse in the context of graphs of bounded treedepth~\cite{jansenpieterse2018}.

\begin{lemma}
    \label{lemma5}
    Let $G$ be a graph and $\eta$ an integer such that $\ed_{\mathcal G_F}(G) \le \eta$.
    Let $S, T \subseteq V(G)$.
    Then one can compute $\FVS(G)$ in time $n^{\bigO_\eta(1)}$.
    Furthermore, one can determine whether there exists a minimum feedback vertex set in $G$ that contains $S$ and separates $T$ in time $n^{\bigO_\eta(1)}$.
\end{lemma}
\begin{proof}[Proof sketch]
    As the approach is standard while working out the details would be tedious, we only sketch the high-level ideas for completeness.
    Observe that by Proposition~\ref{prop:ed-properties}.\ref{prop:tw}, $G$ also has treewidth at most $\eta+1$ and we can determine a tree decomposition of width $\eta+1$ in polynomial time for fixed $\eta$~\cite{treedecompositions}.
    We let $\left(R, (B_u)_{u \in V(R)}\right)$ denote such a tree decomposition.
    The size of a minimum feedback vertex set can be obtained with dynamic programming over this tree decomposition.
    The subproblem for node $t \in V(R)$ with $Z \subseteq B_t$ and a partition $\mathcal P$ of $B_t - Z$ can be defined as the size of a minimum feedback vertex set $Y$ in $G_t$, such that
    \begin{itemize}
        \item $Y \cap B_t$ equals $Z$.
        \item Two vertices of $B_t$ are in the same set in $\mathcal P$ if and only if they are in the same connected component in $G_t - Y$.
    \end{itemize}
    Since $\abs{B_t} \le \eta+2$, the number of partitions can be bounded by $(\eta+2)^{\eta+2}$, which will yield a polynomial running time for constant $\eta$.
    We refer to~\cite[Section 7.3.3]{parameterizedAlgorithms} for a description of the recurrences for these type of problems.

    With a minor modification, we can also derive the size of a minimum set of vertices $Y$ in $G$ such that $Y$ is a feedback vertex set, $Y$ contains $S$ and $Y$ separates $T$ in $G$.
    Observe that the minimum size of such a set is equal to the size of a minimum feedback vertex set, if and only if there exists a minimum feedback vertex set that also contains $S$ and separates $T$.
    After removing $S$ from the graph, we use a similar dynamic programming approach over a tree decomposition.
    The subproblem for node $t \in V(R)$ with $Z \subseteq B_t$ and a partition $\mathcal P$ of $B_t - Z$ now also takes a boolean value for each class of the partition and asks for the minimum size of a feedback vertex set $Y$ in $G_t$ that separates $T \cap V(G_t)$ such that
    \begin{itemize}
        \item $Y \cap B_t$ equals $Z$.
        \item Two vertices of $B_t$ are in the same set in $\mathcal P$ if and only if they are in the same connected component in $G_t - Y$.
        \item A connected component of $G_t - Y$ contains no vertices in $T$ if the boolean value for the corresponding class of the partition is \textsc{False}.
        \item A connected component of $G_t - Y$ contains precisely one vertex in $T$ if the boolean value for the corresponding class of the partition is \textsc{True}.
    \end{itemize}
    When merging two subtrees, we must ensure that a connected component contains only one vertex in $T$ after merging.
    We leave the details of the recursion to the reader.
\end{proof}

Lastly, we will use the concept of a \emph{vertex multiway cut}~\cite{bergougnoux2021node}.
Given a graph $G$ and a set of vertices $T \subseteq V(G)$ (often called the \emph{terminals}), the minimum vertex multiway cut problem asks for a minimum set $S \subseteq V(G)$ such that no connected component of $G-S$ contains more than one terminal.
We say that such a set $S$ \emph{separates} $T$.
It is important that this vertex set $S$ is allowed to intersect the set $T$.
Notice that the edge multiway cut problem is a more common variant, which is for example explained in detail by Garg, Vazirani and Yannakakis~\cite{gargnaveenvizirani1993}.
We will explore a special case of covering/packing duality on trees regarding these cuts, for which we use the following definition.

\begin{definition}[$T$-path]
    Let $G$ be a graph and let $T \subseteq V(G)$ be a set of vertices.
    A $T$-\emph{path} is a path in $G$ whose endpoints are two distinct vertices in $T$.
\end{definition}

If for a graph $G$ with $T \subseteq V(G)$ there exists a collection of $k$ vertex-disjoint $T$-paths, then any vertex multiway cut of $T$ in $G$ has size at least $k$.
While the converse is not true in general, it does hold on a tree.

\begin{proposition}
    \label{terminal-separation-duality}
    Let $G$ be a tree and let $T \subseteq V(G)$ be a set of terminals.
    If $Z \subseteq V(G)$ is a minimum vertex multiway cut of $T$ with size $k$, i.e.\ in $G-Z$ all vertices in $T \setminus Z$ are in different connected components, then there exists a collection of $k$ vertex-disjoint $T$-paths.
\end{proposition}
\begin{proof}
    For a vertex $v$ in a rooted tree, we use $G_v$ to denote the subgraph of $G$ induced by all descendants of $v$, similar to elimination trees.
    Pick a root $r$ of the tree arbitrarily.
    We will assume that $Z$ has the following property:
    for any $u \in Z$ (other than the root) with parent $w$, the set $(Z \setminus \set u )\cup \set w$ is not a vertex multiway cut of $T$.
    Observe that we can construct such a vertex multiway cut greedily from any minimum vertex multiway cut by repeatedly replacing a vertex by its parent when the result still separates $T$ in $G$.

    We will now use induction on $k$, the size of a minimum vertex multiway cut.
    If $k = 1$, then observe that $\abs T \ge 2$ as otherwise the empty set suffices to separate $T$.
    Therefore, we can identify a $T$-path in the tree.
    Now suppose that $k > 1$ and suppose that the claim holds for any tree with a minimum vertex multiway cut of size smaller than $k$.
    Take a tree $G$ that has a minimum vertex multiway cut $Z$ of size $k$.
    Find a vertex $u \in Z$ such that $G_u \cap Z = \set u$, i.e.\ there are no other vertices from the vertex multiway cut in the subtree rooted at $u$.
    Observe that we can always find such a vertex that is not the root of the tree.
    Then we can identify a $T$-path between two terminals in $G_u$:
    replacing $u$ by its parent does not yield a vertex multiway cut, so there are at least two terminals in $G_u$.
    Observe that the $T$-path is contained in $G_u$ as well.
    Now let $G'$ be the graph where $G_u$ is removed from $G$.
    Observe that it is a tree and that a minimum vertex multiway cut has size $k - 1$: the set $Z \setminus \set u$ is a valid vertex multiway cut on $G'$ and if a smaller vertex multiway cut $Z'$ would exist on $G'$, then $Z' \cup \set u$ would have been a smaller vertex multiway cut on $G$.
    Therefore, we can apply the induction hypothesis to obtain a collection of $k - 1$ vertex-disjoint $T$-paths in $G'$.
    Together with the $T$-path that we identified in $G_u$, we obtain a collection of $k$ vertex-disjoint paths in $G$, proving the induction step.
\end{proof}

\subsection{Proof of main lemma}\label{subsec:main-lemma-proof}
We repeat the formulation of Lemma~\ref{lemma3} as stated in the introduction.

\begin{customlemma}{1}
Let $G$ be a connected graph with disjoint vertex sets $S,T \subseteq V(G)$.
Suppose that any minimum feedback vertex set $X$ of $G$ misses some vertex from $S$ or leaves two vertices from $T$ connected in $G-X$.
Then there exist sets $S^* \subseteq S$ and $T^* \subseteq T$ whose sizes only depend on the elimination distance to a forest of $G$, such that any minimum feedback vertex set $X$ of $G$ misses some vertex from $S^*$ or leaves two vertices from $T^*$ connected in $G-X$.
\end{customlemma}

We can split up Lemma~\ref{lemma3} into two parts.
Lemma~\ref{lemma3a} will bound the number of vertices in the $\gf$-elimination tree that contain a vertex in $S$ or $T$.
This part corresponds to the original treedepth formulation in~\cite[Lemma 3]{jansenpieterse2018}, but is significantly simplified for our restricted setting.
Lemma~\ref{lemma3b} bounds the number of vertices in $S$ and $T$ in a bag of the elimination tree.

\begin{lemma}
    \label{lemma3a}
    Let $G$ be a connected graph with disjoint vertex sets $S,T \subseteq V(G)$.
    Let $\left(R, (B_u)_{u \in V(R)}\right)$ be a $\gf$-elimination tree of $G$ of height $\eta$.
    Suppose that any minimum feedback vertex set $X$ of $G$ misses a vertex from $S$ or leaves two vertices from $T$ connected.
    Then this also holds for some subsets $S^* \subseteq S$ and $T^* \subseteq T$, such that any vertex in the elimination tree has at most $3\eta \cdot 2^\eta$ children $u$ for which $G_u$ contains a vertex from $S^*$ or $T^*$.
\end{lemma}

\begin{lemma}
    \label{lemma3b}
    Let $G$ be a connected graph with disjoint vertex sets $S,T \subseteq V(G)$.
    Let $\left(R, (B_u)_{u \in V(R)}\right)$ be a $\gf$-elimination tree of $G$ of height $\eta$.
    Suppose that any minimum feedback vertex set $X$ of $G$ misses a vertex from $S$ or leaves two vertices from $T$ connected.
    Then this also holds for some subsets $S^* \subseteq S$ and $T^* \subseteq T$, such that for any leaf $u$ in the elimination tree, the set $B_u$ contains at most $\eta$ vertices from $S^*$ and at most $\bigO(\eta^2)$ from $T^*$.
\end{lemma}

Observe that these two lemmas directly imply Lemma~\ref{lemma3}.

\begin{proof}[Proof of Lemma~\ref{lemma3}]
    Let $G$ be a graph with vertex sets $S$ and $T$ as stated in Lemma~\ref{lemma3} and let $\left(R, (B_u)_{u \in V(R)}\right)$ be a $\mathcal G_F$-elimination tree of $G$.
    We can first apply Lemma~\ref{lemma3a} to obtain sets $S'$ and $T'$ satisfying the guarantee of that lemma: for each vertex in the elimination tree $R$, we can bound the number of children $u$ for which $B_u$ contains a vertex in $S'$ or $T'$.
    Now define subgraph $R'$ of elimination tree $R$ to be the induced subgraph of $R$ by precisely those vertices $u \in V(R)$ for which $G_u$ contains some vertex in $S'$ or $T'$.
    If $R'$ is the empty graph, then $S'$ and $T'$ are empty, thereby directly implying Lemma~\ref{lemma3}.
    Otherwise, $R'$ is still a rooted tree with height at most $\eta$.
    By Lemma~\ref{lemma3a}, each vertex in this subtree has at most $3 \eta \cdot 2^\eta$ children.
    Since the height of the tree is at most $\eta$, there are at most $(3\eta \cdot 2^\eta)^\eta$ vertices in each layer of $R'$, thereby bounding the number of vertices in $u \in V(R)$ for which $B_u$ contains a labeled vertex.
    However, for leaves $u$ of the elimination tree, there is no bound on the number of vertices in $B_u \cap (S' \cup T')$ yet.
    By applying Lemma~\ref{lemma3b} now on graph $G$ and sets $S'$ and $T'$, we obtain sets $S^* \subseteq S'$ and $T^* \subseteq T'$ such that for each vertex $u$ in elimination tree $R$, the number of vertices in $B_u \cap (S^* \cup T^*)$ is bounded by a function of $\eta$.
    By bounding the total number of labeled vertices by a function of $\eta$, this concludes the proof of Lemma~\ref{lemma3}.
\end{proof}

We now prove Lemma~\ref{lemma3a} and Lemma~\ref{lemma3b}.

\begin{proof}[Proof of Lemma~\ref{lemma3a}]
    In analogy to the original formulation in~\cite{jansenpieterse2018}, a \emph{labeled vertex} is a vertex in $S$ or $T$.
    When we remove a label from a vertex, we remove the vertex from $S$ and $T$ while the vertex remains in the graph.
    Suppose that we pick $S^*$ and $T^*$ such that no minimum feedback vertex set contains $S^*$ and separates $T^*$, while the latter property does not hold for any other pair of sets $S', T'$ with $S' \subseteq S^*$ and $T' \subseteq T^*$.
    We claim that for such sets $S^*$ and $T^*$, any vertex in the elimination tree $R$ has at most $3\eta \cdot 2^\eta$ children $u$ for which $G_u$ contains a labeled vertex.
    We will also refer to the set $T^*$ as the set of terminals.

    Assume towards a contradiction that vertex $v$ has more child subtrees with labels.
    Let these children be $c_1, \dots, c_\ell$.
    For each of these children $c_i$, there exists a minimum feedback vertex set $X_i$ in $G$ that witnesses the fact that the labels cannot be removed from $G_{c_i}$.
    This set $X_i$ will therefore miss a vertex in $S^* \cap G_{c_i}$ or leave a vertex in $T^* \cap G_{c_i}$ connected to some other vertex in $T^*$, while $X_i$ contains all vertices in $S^* \setminus G_{c_i}$ and separates all vertices in $T^* \setminus G_{c_i}$.
    Define $Z_i \coloneqq \tail[v] \setminus X_i$.

    Now fix a set $Z \subseteq \tail[v]$.
    We will bound the number of children $c_i$ for which $Z_i = Z$ by $3 \eta$.
    Suppose towards a contradiction that there are $3 \eta + 1$ of these children.
    Let $C$ be the set containing these vertices.
    Pick some child $c_j \in C$ and observe the following.
    \begin{itemize}
        \item By Proposition~\ref{prop:ed-properties}.\ref{prop:X-cap-Gc}, there are at most $\eta$ children $c_i \in C$ where $X_j \cap G_{c_i}$ is not a minimum feedback vertex set in $G_{c_i}$.
        \item There are at most $\eta$ children $c_i \in C$ with $i \neq j$ such that a terminal in $G_{c_i}$ is connected to a vertex in $Z$ in $G^+_{c_i} - X_j$, i.e.\ \bmp{(recall the notation in Section~\ref{subsec:elim-distances})} in the induced subgraph on the remaining vertices in the subtree and tail of $c_i$.
        Otherwise, two children other than $c_j$ have a terminal connected to the same vertex in $Z$, while $X_j$ separates all terminals outside $G_{c_j}$.
        \item There are at most $\eta - 1$ children $c_i \in C$ such that in $G_{c_i}^+ - X_j$, there exists a path between distinct vertices in $Z$ that uses some vertex in $G_{c_i}$.
        Otherwise, we claim that we can directly construct a cycle in $G - X_j$.
        Consider for example the auxiliary (multi)graph on vertex set~$Z$ which contains, for each child~$c_i \in C$ for which~$G_{c_i}^+ - X_j$ contains such a path, say between~$z_1, z_2 \in Z$, one edge~$z_1 z_2$.
        This auxiliary graph contains a cycle since it has too many edges to be acyclic, which implies that there exists a cycle in $G - X_j$.
    \end{itemize}
    Pick a child $c_k \in C$ which is neither $c_j$ nor in the list of $3\eta -1 $ children above.
    As $\abs C > 3\eta$, such a vertex exists.
    By the first item above, we can deduce that $X_j \cap G_{c_k}$ is a minimum feedback vertex set in $G_{c_k}$.
    Besides, \bmp{this} set contains $S^* \cap G_{c_k}$, it separates all terminals in $T^* \cap G_{c_k}$, and it separates $T^* \cap G_{c_k}$ from $Z$.
    Furthermore, no path exists in $G_{c_k}^+ - X_j$ that connects two vertices in $Z$ and also contains some vertex in $G_{c_k}$.

    \begin{claim}\label{claim:x-prime}
        The set $X' \coloneqq (X_k \setminus G_{c_k}) \cup (X_j \cap G_{c_k})$ is a \bmp{minimum} feedback vertex set in $G$ which contains $S^*$ and separates $T^*$.
    \end{claim}
    \begin{innerproof}
    First, observe that any cycle that is not trivially broken must contain vertices both in $G_{c_k}$ and outside $G_{c_k}$.
    However, when such a cycle leaves $G_{c_k}$, it must enter $Z$.
    This cycle cannot contain only one vertex in $Z$, as it would then also be a cycle in $G_{c_k}^+ - X_j$.
    As no pair of vertices in $Z$ is connected by a path containing some vertex in $G_{c_k}$, such a cycle cannot exist at all.
    Since $X_j \cap G_{c_k}$ is a minimum feedback vertex set in $G_{c_k}$, it follows that $X'$ must be a minimum feedback vertex set in $G$.

    Furthermore, $X_j \cap G_{c_k}$ contains $S^* \cap G_{c_k}$ and separates $T^* \cap G_{c_k}$.
    As $X_k \setminus G_{c_k}$ does the same for the remainder of $S^*$ and $T^*$, we only need to verify that no terminal in $G_{c_k}$ is connected to a terminal outside $G_{c_k}$.
    Because $T^* \cap G_{c_k}$ is also separated from $\tail[v]$, this cannot happen.
    Therefore, $X'$ is a minimum feedback vertex set containing $S^*$ and separating $T^*$.
\end{innerproof}

    Claim~\ref{claim:x-prime} contradicts that any minimum feedback vertex set in $G$ misses a vertex in $S^*$ or leaves two vertices in $T^*$ connected.
    We conclude that there are at most $3\eta$ children $c_i$ of $v$ for which a witnessing minimum feedback vertex set has $Z_i = Z$.
    As there are at most $2^\eta$ subsets of $\tail[v]$ for any non-leaf $v$, this leads to the bound of at most $3 \eta \cdot 2^\eta$ children for which the labels cannot be removed.
\end{proof}

It remains to prove Lemma~\ref{lemma3b}.

\begin{proof}[Proof of Lemma~\ref{lemma3b}]
    Pick some leaf $v$ of elimination tree $R$, for which we want to ensure that there are $\bigO(\eta^2)$ vertices with labels among vertices in $Y \coloneqq B_v$.
    Define $S_Y \coloneqq S \cap Y$ and $T_Y \coloneqq T \cap Y$.
    Our goal is to obtain subsets $S_Y^* \subseteq S_Y$ and $T_Y^* \subseteq T_Y$ whose sizes are $\bigO(\eta^2)$, such that every minimum feedback vertex set misses a vertex from $S_Y^* \cup (S \setminus Y)$ or leaves a pair of terminals in $T_Y^* \cup (T \setminus Y)$ connected.
    By applying this operation to all leaves of the elimination tree, we obtain the sets promised by Lemma~\ref{lemma3b}.

    The construction of $S_Y^*$ is straightforward.
    If $\abs{S_Y} > \eta + 1$, we let $S_Y^*$ be an arbitrary subset of $S_Y$ of size $\eta + 1$.
    Otherwise, $S_Y^* = S_Y$.
    
    \begin{claim}
        Let $X$ be a minimum feedback vertex set in $G$.
        If $X$ misses a vertex in $S$, then it also misses a vertex in $S_Y^* \cup (S \setminus Y)$.
    \end{claim}
    \begin{innerproof}
        If $X$ misses a vertex in $S \setminus Y$, then the implication is trivial.
        Therefore assume $X$ misses a vertex in $S_Y$.
        If this vertex is not in $S_Y^*$, then $\abs{S_Y^*} = \eta + 1$ by construction.
        By Proposition~\ref{prop:ed-properties}.\ref{prop:X-cap-Bv}, we know that $\abs{X \cap S_Y^*} \le \eta$ so $X$ misses a vertex in $S_Y^*$.
    \end{innerproof}
    
    For the construction of $T_Y^*$ we distinguish two cases.
    First, we assume that $T_Y$ cannot be separated with $\eta$ vertices in $G[Y]$ and \bmp{make} the following observation.
    
    \begin{proposition}\label{prop:disjoint-paths}
        Let $G$ be a tree and $T \subseteq V(G)$.
        If $T$ cannot be separated with $\eta$ vertices, then there exist $\eta + 1$ vertex-disjoint paths whose endpoints are distinct vertices in $T$.
    \end{proposition}
    \begin{innerproof}
    Let $Z$ be a minimum vertex multiway cut of $T$ in $G$, then $\abs Z \ge \eta + 1$.
    It remains to apply Proposition~\ref{terminal-separation-duality} to obtain $\eta + 1$ vertex-disjoint paths whose endpoints are distinct vertices in $T$.
\end{innerproof}
    
    We define $T_Y^*$ by taking the $2\eta + 2$ endpoints of the paths guaranteed by Proposition~\ref{prop:disjoint-paths}.
    Observe that these vertices are all in $T_Y$.
    
    \begin{claim}
        Suppose that $T_Y$ cannot be separated with $\eta$ vertices in $G[Y]$.
        Let $X$ be a minimum feedback vertex set in $G$.
        Then $X$ leaves two vertices in $T_Y^*$ connected.
    \end{claim}
    \begin{innerproof}
        By Proposition~\ref{prop:ed-properties}.\ref{prop:X-cap-Bv}, $X$ can only intersect $\eta$ of the $\eta + 1$ vertex-disjoint paths that were obtained through Proposition~\ref{prop:disjoint-paths}.
        Therefore, at least one path is disjoint from $X$, so its endpoints in $T_Y^*$ are connected in $G-X$.
    \end{innerproof}
    
    It remains to consider the case where $T_Y$ can be separated with $\eta$ vertices.
    Let $Z$ be a vertex multiway cut of $T_Y$ in $G[Y]$ with $\abs Z \le \eta$ and let $\mathcal C \coloneqq \cc(G[Y] - Z)$.
    Observe that each of these connected components is a tree with at most one vertex in $T_Y$.
    Let $\mathcal C_T \subseteq \mathcal C$ be the set of components that contain a vertex in $T_Y$.
    We are now going to mark components.
    For each $z \in Z$, mark $\eta + 2$ components in $\mathcal C_T$ that are adjacent to $z$ in $G[Y]$, or all if there are fewer.
    Similarly, for each $u \in \tail(v)$ we mark up to $\eta + 2$ components in $\mathcal C_T$ that are adjacent to $v$ in $G_v^+$.
    Then we define $T_Y^*$ to be the union of all vertices in $T_Y$ in the marked components, together with $Z \cap T_Y$.
    These are at most $\eta (\eta + 2) + \eta (\eta + 2) + \eta = \bigO(\eta^2)$ vertices.
    
    \begin{claim}
        Suppose that $T_Y$ can be separated with $\eta$ vertices in $G[Y]$.
        Let $X$ be a minimum feedback vertex set in $G$ and suppose that $X$ leaves two vertices in $T$ connected.
        Then $X$ also leaves two vertices in $T_Y^* \cup (T \setminus Y)$ connected.
    \end{claim}
    \begin{innerproof}
        Let $Z$ be the vertex multiway cut used in the construction of $T_Y^*$ and let $t_1, t_2 \in T$ be two terminals that are connected in $G - X$.
        If they are both in $T_Y^* \cup (T \setminus Y)$, then the implication is trivial, so assume that $t_1 \in T_Y$ but not in $T_Y^*$.
        Observe that therefore $t_1 \not \in Z$.
        Let $P$ be a path from $t_1$ to $t_2$ in $G - X$ and let $z$ be the first vertex on this path that is not in $G[Y] - Z$.
        We now distinguish \bmp{two} cases.
        If $z \in Z$, then observe that $t_1$ was in a component in $\mathcal C_T$ that was not marked.
        Then there are $\eta + 2$ marked components in $\mathcal C_T$ adjacent to $z$ in $G[Y]$ of which the terminals are in $T_Y^*$.
        Only $\eta$ of these components can be intersected by $X$ by Proposition~\ref{prop:ed-properties}.\ref{prop:X-cap-Bv}, so there exists a path between two terminals in $T_Y^*$ in $G[Y]$.
        If $z \not \in Z$, then we obtain that $z \in \tail(v)$ by Proposition~\ref{prop:ed-properties}.\ref{prop:path-via-tail} and the case follows analogously.
    \end{innerproof}

    This concludes the construction of the sets $S_Y^*$ and $T_Y^*$.
    If any minimum feedback vertex set in $G$ misses a vertex in $S$ or leaves a pair of terminals in $T$ connected, then it also misses a vertex in $S_Y^* \cup (S \setminus Y)$ or leaves a pair of terminals in $T_Y^* \cup (T \setminus Y)$ connected.
    By applying this operation to all leaves of the elimination tree, we obtain the promised sets $S^*$ and $T^*$ which concludes the proof of Lemma~\ref{lemma3b}.
\end{proof}

With Lemma~\ref{lemma3a} and Lemma~\ref{lemma3b} proven, we conclude the proof of Lemma~\ref{lemma3}: if any minimum feedback vertex set in a graph $G$ misses a vertex from a set $S \subseteq V(G)$ or leaves two terminals in a set $T \subseteq V(G)$ connected, then this property also holds for sets $S^* \subseteq S$ and $T^* \subseteq T$ whose sizes only depend on $\ed_{\gf}(G)$.

\subsection{Kernelization algorithm}\label{subsec:kernel}

Our kernel follows the structure of the polynomial kernel for $\mathcal F$-\prob{Minor Free Deletion} when parameterized by a treedepth-$\eta$ modulator for some integer $\eta$~\cite{jansenpieterse2018}.
Our kernel relies crucially on the reduction rule specified in Lemma~\ref{lemma6}, whose correctness relies on Lemma~\ref{lemma3}.

\begin{lemma}[{Cf.~\cite[Lemma 6]{jansenpieterse2018}}]\label{lemma6}
There is a polynomial-time algorithm that, given a graph $G$ with modulator $X \subseteq V(G)$ such that $\ed_{\gf}(G - X) \le \eta$ for a constant $\eta$, outputs an induced subgraph $G'$ of $G$ together with an integer $\Delta$ such that $\FVS(G) = \FVS(G') + \Delta$ and $G' - X$ has at most $\abs{X}^{\bigO_\eta(1)}$ components.
\end{lemma}

We can use this reduction rule to obtain a graph $G'$ where $G' - X$ has a bounded number of connected components.
We can then identify a set of vertices $Y \subseteq V(G'-X)$ with $\abs Y \le \abs{X}^{\bigO_\eta(1)}$ such that $\ed_{\gf}(G' - X - Y) < \eta$.
By definition of elimination distance, every connected component $C$ of $G'-X$ contains a vertex whose removal decreases $\ed_{\gf}(C)$.
As we limited the number of connected components by applying Lemma~\ref{lemma6}, these vertices constitute a suitable set $Y$.
Now observe that $X \cup Y$ is a modulator to a graph with elimination distance to a forest $\eta - 1$ and that $\abs{X \cup Y}$ is bounded by a polynomial in $\abs X$.
One can therefore provide an inductive argument which repeatedly applies Lemma~\ref{lemma6} and increases the modulator such that the elimination distance to a forest of the remaining graph decreases every iteration.
Once we obtain a modulator to a graph with elimination distance to a forest 1 (which is a forest), we can apply a known polynomial kernel in the size of a feedback vertex set~\cite{fvsbetterquadratickernel}.
We formalize these ideas in Theorem~\ref{thm:main-kernel}.

\begin{theorem}
    \label{thm:main-kernel}
    Let $\eta \ge 0$ be an integer.
    Let $\mathcal G_\eta$ be the set of graphs with elimination distance to a forest at most $\eta$.
    Then \prob{Feedback Vertex Set} admits a polynomial kernel in the size of a $\mathcal G_{\eta}$-modulator $X$.
\end{theorem}

For our kernel, we will assume that such a $\mathcal G_{\eta}$-modulator is provided together with the input.
This has the technical reason that we otherwise cannot verify efficiently whether such a modulator actually exists.
In other words, we cannot determine efficiently for a graph $G$ with integers $k$ and $\ell$ whether there exists a feedback vertex set of size $\ell$ and a $\mathcal G_{\eta}$-modulator of size $k$ as the latter is an NP-hard problem.
By providing the modulator $X$ with the parameterized instance, we only need to verify that $X$ is indeed such a modulator.
This can be done efficiently when $G-X$ has bounded elimination distance to a forest, since determining elimination distances is fixed parameter tractable in the target value~\cite{eldist}.
In practice, one can determine a $\log(c+1)$-approximation for the deletion distance to a minor-closed graph family with elimination distance to a forest bounded by $c$~\cite{treewidthapproximation}, since the treewidth is bounded as well by Proposition~\ref{prop:ed-properties}.\ref{prop:tw}.
This approximation with an effectively constant approximation ratio improves the $\abs{X_{\text{OPT}}} \log^{3/2} \abs{X_{\text{OPT}}}$-approximation which was used in earlier papers~\cite{approxalgforbiddenminors,bridgedepth}.
Here, $X_{\text{OPT}}$ denotes a minimum-sized modulator.

\begin{proof}[Proof of Theorem~\ref{thm:main-kernel}]
    Recall that $X$ is a subset of $V(G)$ such that $\ed_{\mathcal G_F}(G - X) \le \eta$.
    We will prove the theorem using induction on $\eta$.

    If $\eta = 0$ then $G - X$ is a forest.
    It is known that feedback vertex set admits a quadratic kernel in the size of a solution~\cite[Section 9.1]{parameterizedAlgorithms}, concluding the base case.

    If $\eta > 0$, we apply Lemma~\ref{lemma6} to obtain a graph $G'$ and an integer $\Delta$ such that $G' - X$ has at most $\abs X^{\bigO_\eta(1)}$ connected components and $\FVS(G') + \Delta = \FVS(G)$.
    Observe that $G'$ and $\Delta$ can be obtained in polynomial time for fixed $\eta$.
    We will now extend the modulator $X$ to a modulator $X'$ such that $G' - X'$ has $\ed_{\mathcal G_F}(G' - X') < \eta$.
    Consider each connected component $C$ of $G' - X$ for which $\ed_{\mathcal G_F}(C) = \eta$.
    By definition, for each such component $C$ there must exist a vertex $v \in V(C)$ such that $\ed_{\mathcal G_F}(C - \set v) < \eta$.
    We can find such a vertex by trying all options for $v$ and determining $\ed_{\mathcal G_F}(C - \set v)$.
    This can be done in polynomial time for constant $\eta$ since determining elimination distances to minor-closed graph classes is fixed parameter tractable in the target value~\cite{eldist}.
    By adding this vertex to $X$ for each considered component $C$, we obtain a set $X'$ whose size is still bounded by a polynomial in $\abs X$ for constant $\eta$ since the number of components to consider is polynomial in $\abs X$.
    Since each connected component $C$ of $G' - X'$ has $\ed_{\mathcal G_F}(C) < \eta$, it follows that $\ed_{\mathcal G_F}(G' - X') < \eta$.
    Then by induction, the \prob{Feedback Vertex Set} problem on $G'$ admits a polynomial kernel in the size of $X'$, which is a modulator to a graph with elimination distance to a forest at most $\eta - 1$.
    Since $\FVS(G') + \Delta = \FVS(G)$ by Lemma~\ref{lemma6}, we know that $G'$ has a feedback vertex set of size at most $\ell$ if and only if $G$ has a feedback vertex set of size at most $\ell + \Delta$.
    This concludes the induction step.
\end{proof}

It remains to prove Lemma~\ref{lemma6}.
\begin{proof}[Proof of Lemma~\ref{lemma6}]
    Pick $\gamma \in \bigO_\eta(1)$ such that for any graph with elimination distance to a forest at most $\eta$ and vertex sets $S$ and $T$, Lemma~\ref{lemma3} returns subsets $S^* \subseteq S$ and $T^* \subseteq T$ of size at most $\gamma$.
    Define $\tau = \abs X + 1 + 3 \gamma$ and $\mathcal C = \cc(G-X)$.

    We are now going to add components from $\mathcal C$ to a set $\mathcal C'$.
    We will ensure that the size of $\mathcal C'$ is polynomial in $\abs X$ for constant $\eta$ and use this set to prove the lemma.
    Consider each set $X_0 \subseteq X$ with $\abs{X_0} \le 3\gamma$.
    Observe that these are $\bigO(\abs X^{3 \gamma}) = \abs X^{\bigO_\eta(1)}$ subsets to consider.
    For each of these sets $X_0 \subseteq X$, we do the following for each component $C \in \mathcal C$:
    we consider the sets $\mathcal S(C, X_0)$ and $\mathcal T(C, X_0)$ as defined in Definition~\ref{def:st}.
    We then determine whether there exists a minimum feedback vertex set in $C$ containing $\mathcal S(C, X_0)$ and separating $\mathcal T(C, X_0)$.
    By Lemma~\ref{lemma5}, this can be verified in polynomial time in graphs with bounded elimination distance to a forest.
    If there are at most $\tau$ components in $\mathcal C$ where such a minimum feedback vertex set does not exist, then we add all components to $\mathcal C'$;
    otherwise we pick $\tau$ of them.
    Observe that $\abs {\mathcal C'} = \abs X^{\bigO_\eta(1)}$, since we add at most $\tau$ components for each considered subset of $X$.

    We now claim the following regarding components that are not in $\mathcal C'$.
    \begin{claim}\label{claim:lemma6}
    For any component $C^* \in \mathcal C \setminus \mathcal C'$, $\FVS(G) = \FVS(G - V(C^*)) + \FVS(C^*)$.
    \end{claim}
    \begin{innerproof}
        Let $\widehat G = G - V(C^*)$.
        Let $Y$ be a minimum feedback vertex set in $G$.
        Then $Y$ can be split into a feedback vertex set $Y_1$ in $\widehat G$ and a feedback vertex set $Y_2$ in $C^*$.
        It follows that $\FVS(G) = \abs Y = \abs{Y_1} + \abs{Y_2} \ge \FVS(\widehat G) + \FVS(C^*)$, and it remains to prove the other inequality.

        Let $\widehat Y$ be a minimum feedback vertex set in $\widehat G$ and let $X_0 = X \setminus \widehat Y$.
        We again consider sets $\mathcal S(C^*, X_0)$ and $\mathcal T(C^*, X_0)$.
        If there exists a minimum feedback vertex set $Y^*$ in $C^*$ that contains $\mathcal S(C^*, X_0)$ and separates $\mathcal T(C^*, X_0)$, then by Proposition~\ref{prop:samen-fvs}, $Y^* \cup \widehat Y$ is a minimum feedback vertex set in $G$.
        It follows that $\FVS(G) \le \FVS(\widehat G) + \FVS(C^*)$ under these circumstances.
        Therefore, suppose towards a contradiction that each minimum feedback vertex set in $C^*$ misses a vertex in $\mathcal S(C^*, X_0)$ or leaves a pair of vertices in $\mathcal T(C^*, X_0)$ connected.
        Then by Lemma~\ref{lemma3}, this also holds for sets $S^* \subseteq \mathcal S(C^*, X_0)$ and $T^* \subseteq \mathcal T(C^*, X_0)$ whose sizes are bounded by $\gamma$.

        Now consider the following set $X^* \subseteq X_0$:
        for each vertex in $S^*$, add two arbitrary neighbors in $X_0$ to $X^*$.
        For each vertex in $T^*$, add its neighbor in $X_0$ to $X^*$.
        Now consider the sets $\mathcal S(C^*, X^*)$ and $\mathcal T(C^*, X^*)$.
        Observe that $S^* \subseteq \mathcal S(C^*, X^*)$ and $T^* \subseteq \mathcal T(C^*, X^*)$.
        It follows that any minimum feedback vertex set in $C^*$ misses a vertex in $\mathcal S(C^*, X^*)$ or leaves a pair of vertices in $\mathcal T(C^*, X^*)$ connected.
        Combined with the fact that $\abs{X^*} \le 3 \gamma$, it means that $C^*$ could have been added to $\mathcal C'$ during its construction when we considered $X^*$ as a subset of $X$.
        However, the component was not added, implying that $\mathcal C'$ contains at least $\tau$ components $C$ where every minimum feedback vertex set misses a vertex from $\mathcal S(C, X^*)$ or leaves a pair of vertices in $\mathcal T(C, X^*)$ connected.
        Let $\mathcal C''$ be the set of these components.
        By Proposition~\ref{prop:max-gamma-components}, there can be at most $3\gamma$ components $C \in \mathcal C''$ where $\widehat Y \cap V(C)$ misses a vertex from $\mathcal S(C, X^*)$ or leaves a pair of vertices in $\mathcal T(C, X^*)$ connected in $C - \widehat Y$.
        Since $\tau = \abs X + 1 + 3\gamma$, there are at least $\abs X + 1$ components $C \in \mathcal C''$ where $\widehat Y \cap V(C)$ contains $\mathcal S(C, X^*)$ and separates $\mathcal T(C, X^*)$, implying that $\widehat Y \cap V(C)$ is not a minimum feedback vertex set in $C$.
        We will now construct a different feedback vertex set $\widehat Y'$ on $\widehat G$:
        \begin{itemize}
            \item For each of the at least $\abs X + 1$ components $C \in \mathcal C \setminus \set{C^*}$ where $\widehat Y \cap V(C)$ is not a minimum feedback vertex set of $G[C]$, we add a minimum feedback vertex set in $C$.
            \item For every other component $C \in \mathcal C \setminus \set{C^*}$, we take $\widehat Y \cap V(C)$.
            \item We add $X$ to $\widehat Y'$.
        \end{itemize}
        Observe that $\widehat Y'$ is a feedback vertex set in $\widehat G$: since $X$ is included in $\widehat Y'$, any cycle must be contained in a component in $\mathcal C \setminus \set{C^*}$, but for each of these components we added a feedback vertex set in $C$ to $\widehat Y'$.
        Furthermore, $\abs{\widehat Y'} \le \abs{\widehat Y} - (\abs X + 1) + \abs X < \abs{\widehat Y}$, contradicting the assumption that $\widehat Y$ is a minimum feedback vertex set.

        We conclude that there exists a minimum feedback vertex set in $C^*$ that contains $\mathcal S(C^*, X_0)$ and separates $\mathcal T(C^*, X_0)$, proving that $\FVS(G) \le \FVS(\widehat G) + \FVS(C^*)$ as we saw earlier.
    \end{innerproof}
    We can now complete the proof of Lemma~\ref{lemma6}.
    Let $G'$ be the subgraph of $G$ induced by $X$ and all vertices in components in $\mathcal C'$.
    Observe that the number of components of $G' - X$ is polynomial in $\abs X$.
    Let $\Delta$ be the size of a minimum feedback vertex set in the subgraph consisting of all components in $\mathcal C \setminus \mathcal C'$, which we can obtain in polynomial time by Lemma~\ref{lemma5}.
    By applying Claim~\ref{claim:lemma6} to all components in $\mathcal C \setminus \mathcal C'$, we conclude that $\FVS(G) = \FVS(G') + \Delta$.
\end{proof}

    \section{Kernelization lower bounds}\label{sec:lower-bound}

    We first introduce some relevant terminology and its properties.

\subsection{Additional preliminaries for Section~\ref{sec:lower-bound}}\label{app:prelim-lower-bound}
We sometimes slightly abuse the $\bigO_{\eta}$ notation here and write $\bigO_{\mathcal F}(f(n))$ for a set of graphs $\mathcal F$, rather than writing the size of the largest graph in $\mathcal F$ as subscript.
For a graph $G$ and a set of connected graphs $\mathcal F$, we write $\ed_{\neg \mathcal F}(G)$ to denote the elimination distance to an $\mathcal F$-minor free graph.

We introduce the notion of \bmp{a necklace}, which turns out to be a crucial structure.
\begin{definition}
    Let $G$ be a graph and let $\mathcal F$ be \bmp{a} collection of connected graphs.
    $G$ is an $\mathcal F$-\emph{necklace} of length $t$ if there exists a partition of $V(G)$ into $S_1, \dots, S_t$ such that
    \begin{itemize}
        \item $G[S_i] \in \mathcal F$ for each $i \in [t]$ (these subgraphs are the \emph{beads} of the necklace),
        \item $G$ has precisely one edge between $S_i$ and $S_{i+1}$ for each $i \in [t-1]$,
        \item $G$ has no edges between any other pair of sets $S_i$ and $S_j$.
    \end{itemize}
\end{definition}

When the length of the necklace is not relevant, we simply speak of an $\mathcal F$-necklace.
The following definition specifies a special type of necklace.

\begin{definition}
    \label{def:uniform-necklace}
    Let $\mathcal F$ be a collection of connected graphs.
    Let $G$ be an $\mathcal F$-necklace of length $t$.
    We say that $G$ is a \emph{uniform necklace} if it satisfies two additional conditions.
    \begin{itemize}
        \item There exists a graph $H \in \mathcal F$ such that each bead $G[S_i]$ is isomorphic to $H$.
        \item There exist $x, y \in V(H)$ and graph isomorphisms $f_i \colon V(H) \rightarrow V(G[S_i])$ for each bead $G[S_i]$, such that for each $i \in [t-1]$, the edge between $G[S_i]$ and $G[S_{i+1}]$ has precisely the endpoints $f_i(x)$ and $f_{i+1}(y)$.
    \end{itemize}
\end{definition}
Notice that for a collection of graphs $\mathcal F$, we can specify a uniform necklace uniquely by a graph $H \in \mathcal F$, its length $t$ and the two vertices in $x, y \in V(H)$ indicating the endpoints of the edges between two consecutive beads.
It will be useful to consider all uniform necklaces of different lengths with this structure.
We speak of the \emph{uniform necklace structure} $(H, x, y)$ to describe all uniform $\set{H}$-necklaces where there exist isomorphisms as described in Definition~\ref{def:uniform-necklace}:
for each pair of consecutive beads $S_i$ and $S_{i+1}$ with edge $e$ connecting them, the endpoint of $e$ in $S_i$ equals the image of $x$ under the $i$th isomorphism and the endpoint in $S_{i+1}$ equals the image of $y$ under the $(i+1)$st.

The following proposition explains why we consider this special type of necklaces.
\begin{proposition}
    \label{prop:always-uniform-necklaces}
    Let $\mathcal F$ be a finite collection of connected graphs.
    Let $\mathcal G$ be a minor-closed graph family that contains arbitrarily long $\mathcal F$-necklaces.
    Then $\mathcal G$ also contains arbitrarily long uniform $\mathcal F$-necklaces.
    Moreover, there exists a graph $H \in \mathcal F$ and vertices $x, y \in V(H)$ such that all uniform necklaces with structure $(H, x, y)$ are contained in $\mathcal G$.
\end{proposition}
\begin{proof}
    We first argue that there exists a graph $H \in \mathcal F$ such that $\mathcal G$ contains arbitrarily long $\set H$-necklaces.
    Observe that if only a finite number $c$ of beads in necklaces in $\mathcal G$ can be isomorphic to each graph in $\mathcal F$, then $\mathcal G$ cannot contain necklaces of length more than $c \cdot \abs{\mathcal F}$.
    It follows that there exists a graph $H \in \mathcal F$ such that $\mathcal G$ contains necklaces where arbitrarily many beads are isomorphic to $H$.
    If we take such a necklace, contract all edges in beads which are not isomorphic to $H$ and then contract the resulting paths between beads sufficiently, we can create arbitrarily long $\set H$-necklaces.
    These are in $\mathcal G$ as well since it is a minor-closed graph family.
    We can now take an $\set H$-necklace $G$ of arbitrary length $t$ in $\mathcal G$ with beads $G[S_1], \dots, G[S_t]$.
    Let $f_i \colon V(H) \to S_i$ be a graph isomorphism between $H$ and the $i$th bead for each $i \in [t]$.
    Now consider a bead $G[S_i]$ with $2 \le i \le t-1$ where $u \in S_i$ is the endpoint of the edge to $S_{i-1}$ and $v \in S_i$ is the endpoint of the edge to $S_{i+1}$.
    There are $\abs{V(H)}$ vertices in $H$ that can be the preimage of $u$ under the $i$th isomorphism and, similarly, there are $\abs{V(H)}$ options for $v$.
    Therefore, by the pigeonhole principle there exist vertices $x, y \in V(H)$ such that for at least $\lceil \frac{t-2}{\abs{V(H)}^2} \rceil$ beads $G[S_i]$, the vertex in $S_i$ adjacent with $S_{i-1}$ equals the image of $x$ under the $i$th isomorphism and the vertex in $S_i$ adjacent to $S_{i+1}$ equals the image of $y$ under that isomorphism.
    By contracting the edges within all other beads and by contracting the resulting paths between beads sufficiently, we obtain a uniform $\mathcal F$-necklace of length $\lceil \frac{t-2}{\abs{V(H)}^2} \rceil$ which is contained in $\mathcal G$ as it is minor-closed.
    By picking $t$ arbitrarily large, we conclude that $\mathcal G$ contains arbitrarily long uniform $\mathcal F$-necklaces.

    Furthermore, suppose that for some graph $H$ and vertices $x, y \in V(H)$, $\mathcal G$ does not contain all uniform $\mathcal F$-necklaces with structure $(H, x, y)$.
    Since $\mathcal G$ is minor-closed, this implies that there exists a constant $c$ bounding the length of all uniform $\mathcal F$-necklaces with structure $(H, x, y)$ in $\mathcal G$.
    Now suppose towards a contradiction that for no $H \in \mathcal F$ and $x, y \in V(H)$, $\mathcal G$ contains all uniform $\mathcal F$-necklaces with structure $(H, x, y)$.
    As the number of structures is finite for a finite collection of finite graphs, we can take the maximum of these constants.
    This constant then bounds the length of any uniform necklace in $\mathcal G$, while we proved earlier that $\mathcal G$ contains arbitrarily long uniform necklaces.
    We obtain a contradiction and thereby conclude that there exists an $H \in \mathcal F$ and $x, y \in V(H)$ such that $\mathcal G$ contains all uniform necklaces with structure $(H, x, y)$.
\end{proof}

The following observation follows directly from the definitions of a graph minor.

\begin{observation}
    \label{obs:minor-conn}
    Let $G$ be a graph and $H$ be a connected graph.
    If $\phi$ is a minor model of $H$ in $G$, then the graph $G[\phi(V(H))]$ is connected.
\end{observation}

We will often restrict ourselves to necklaces where each graph $H \in \mathcal F$ is \emph{biconnected}, meaning that $H$ is connected and for any $v \in V(H)$ the graph $H - \set v$ is still connected.
A \emph{biconnected component} of a graph $G$ is a maximal biconnected subgraph of $G$, i.e. it is a biconnected subgraph of $G$ which is not a proper subgraph of any other biconnnected subgraph of $G$.
The following proposition gives a necessary condition for a graph to contain a biconnected graph as a minor.
\begin{proposition}
    \label{prop:biconn}
    Let $H$ be a biconnected graph and let $G$ be a graph which contains $H$ as a minor.
    Then for any minimal minor model $\phi$ of $H$ in $G$, the graph $G[\phi(V(H))]$ is biconnected.
    Furthermore, this graph is a minor of a biconnected component.
\end{proposition}
\begin{proof}
    Suppose that $G$ contains $H$ as a minor and let $\phi$ be a minimal minor model of $H$ in $G$.
    Assume towards a contradiction that $G[\phi(V(H))]$ is not biconnected.
    Let $v \in \phi(V(H))$ be a vertex such that $G[\phi(V(H)) \setminus \set v]$ is not connected.
    Let $u \in V(H)$ be the vertex such that $v \in \phi(u)$.
    Suppose that for some component $C$ of the disconnected graph $G[\phi(V(H)) \setminus \set v]$, its vertex set $V(C)$ is included in $\phi(u)$.
    Observe that any edge of $G[\phi(V(H))]$ with precisely one endpoint in $V(C)$ will have $v$ as its other endpoint, since $C$ is a component of $G[\phi(V(H)) \setminus \set v]$.
    It follows that any edge with at least one endpoint in $C$ is entirely contained in $\phi(u)$ in $G[\phi(V(H))]$.
    Now consider the minor model $\phi'$ of $H$ in $G$, defined by
    \[
        \phi'(w) = \begin{cases}
                       \phi(w) & \text{ if $w \neq u$,} \\
                       \phi(w) \setminus V(C) & \text{ if $w = u$.} \\
        \end{cases}
    \]
    Then observe that $\phi'(V(H))$ is still connected and for any edge $xy$ in $H$, there still exists an edge between $\phi'(x)$ and $\phi'(y)$ in $G[\phi'(V(H))]$, since such an edge does not contain an endpoint in $C$.
    It follows that $\phi$ is not a minimal minor model of $H$ in $G$ in this case, so no component $C$ of $G[\phi(V(H)) \setminus \set v]$ is included in $\phi(u)$.
    We now claim that $H - \set u$ is also not connected.
    To this end, we will consider the subgraph of $G[\phi(V(H))]$ where not only $v$ is removed, but the entire set $\phi(u)$.
    Observe that $G'$ still contains multiple components, since no component of $G[\phi(V(H)) \setminus \set v]$ is contained in $\phi(u)$.
    Now for each $x, y \in V(H) \setminus \set u$ for which $xy \not \in E(H)$, remove all edges between a vertex in $\phi(x)$ and a vertex in $\phi(y)$, and contract all edges between vertices in $\phi(x)$ for each $x \in V(H) \setminus \set u$.
    By definition of a minor model, we obtain the graph $H - \set u$.
    However, after contracting and removing edges in a disconnected graph, the remaining graph is still disconnected.
    Thereby $H - \set u$ is disconnected, so $H$ was not biconnected.
    We conclude that for any minimal minor model $\phi$ of $H$ in $G$, the graph $G[\phi(V(H))]$ is biconnected.
    It follows directly that $G[\phi(V(H))]$ is a subgraph (and thereby also a minor) of a biconnected component of $G$: if it is not a proper subgraph of any biconnected component of $G$, then it is not a proper subgraph of any biconnected subgraph of $G$.
    But then, by definition, $G[\phi(V(H))]$ is a biconnected component of $G$ itself.
\end{proof}

Proposition~\ref{prop:biconn} directly implies that if some graph $G$ contains a biconnected graph $H$ as a minor, then some biconnected component of $G$ contains an $H$-minor.
This provides us with a useful tool for arguing that a necklace does not contain some biconnected graph as a minor.
To this end, observe that we can characterize the biconnected components of a necklace in the following way.
Notice that we do not require the graphs in $\mathcal F$ to be biconnected here.

\begin{observation}
    \label{obs:biconn-comps-necklace}
    Let $\mathcal F$ be a collection of connected graphs and let $G$ be an $\mathcal F$-necklace.
    Then any biconnected component of $G$ with at least three vertices is a subgraph of a bead.
\end{observation}

The observation follows from the fact that for any pair of beads, we can identify an edge that is used on any path connecting those two beads.
Proposition~\ref{prop:biconn} and Observation~\ref{obs:biconn-comps-necklace} imply the following result.

\begin{lemma}
    \label{lemma:necklace-minors}
    Let $\mathcal F$ be a collection of connected graphs and let $G$ be an $\mathcal F$-necklace.
    Let $\mathcal H$ be a collection of biconnected graphs on at least three vertices.
    If no graph in $\mathcal F$ contains an $\mathcal H$-minor, then $G$ contains no $\mathcal H$-minor.
\end{lemma}
\begin{proof}
    Assume towards a contradiction that $G$ contains an $H$-minor for some $H \in \mathcal H$.
    Then by Proposition~\ref{prop:biconn}, $G$ has a biconnected component with an $H$-minor.
    Since $H$ has at least three vertices, this biconnected component must have at least three vertices as well.
    By Observation~\ref{obs:biconn-comps-necklace}, such a biconnected component is a subgraph of a bead of $G$.
    We conclude that $G$ has a bead containing an $H$-minor, contradicting that no graph in $\mathcal F$ contains an $\mathcal H$-minor.
\end{proof}

Lemma~\ref{lemma:necklace-minors} argues that it suffices to hit the $\mathcal H$-minors in each bead of a necklace to ensure that a necklace has no $\mathcal H$-minors at all.
We will use these ideas when $\mathcal F$ is a collection of proper subgraphs of graphs in $\mathcal H$.
Lemma~\ref{lemma:necklace-minors} also implies the following corollary.

\begin{corollary}
    \label{corr:planar}
    Let $\mathcal F$ be a collection of planar graphs and let $G$ be an $\mathcal F$-necklace.
    Then $G$ is also planar.
\end{corollary}
\begin{proof}
    By Kuratowski's celebrated theorem, a graph is planar if and only if it contains no $K_5$ and no $K_{3,3}$ minor.
    Observe that the set $\mathcal H = \set{K_5, K_{3,3}}$ is a set of biconnected graphs on at least three vertices.
    Since no graph in $\mathcal F$ contains an $\mathcal H$-minor by Kuratowski's theorem, also $G$ contains no $\mathcal H$-minor by Lemma~\ref{lemma:necklace-minors}.
    We conclude that $G$ is planar by Kuratowski's theorem.
\end{proof}

\subsection{Proof setup}
We repeat the lower bound we aim to prove here as stated in the introduction.

\begin{customthm}{2}
    Let $\mathcal G$ be a minor-closed family of graphs and let $\mathcal F$ be a finite set of biconnected planar graphs on at least three vertices.
    If $\mathcal G$ has unbounded elimination distance to an $\mathcal F$-minor free graph, then $\mathcal F$-\prob{Minor Free Deletion} does not admit a polynomial kernel in the size of a $\mathcal G$-modulator, unless $\NP \subseteq \coNP / \poly$.
\end{customthm}

Notice that when we consider $\mathcal F$-\prob{Minor Free Deletion} for some collection of graphs $\mathcal F$, we can assume that there do not exist distinct graphs $G, H \in \mathcal F$ such that $G$ is a minor of $H$.
Our proof consists of two parts.
We first derive the following property, where we say that a set contains arbitrarily long necklaces if there does not exist a constant $c$ such that each necklace in the set has length at most $c$.

\begin{lemma}
    \label{lemma:unbounded-ed-implies-unbounded-necklaces}
    Let $\mathcal F$ be a finite collection of connected planar graphs.
    Any minor-closed graph family $\mathcal G$ with unbounded elimination distance to an $\mathcal F$-minor free graph contains arbitrarily long uniform $\mathcal F$-necklaces.
\end{lemma}
Then we will prove the following lemma by giving a reduction from \prob{CNF Satisfiability} parameterized by the number of variables~\cite{dellmelkebeek2010}.
\begin{lemma}
    \label{lemma:unbounded-necklaces-implies-geen-kernel}
    Let $\mathcal F$ be a finite set of biconnected planar graphs on at least three vertices and let $\mathcal G$ be a minor-closed graph family.
    If $\mathcal G$ contains arbitrarily long uniform $\mathcal F$-necklaces, then $\mathcal F$-\prob{Minor Free Deletion} does not admit a polynomial kernel in the size of a $\mathcal G$-modulator, unless $\NP \subseteq \coNP / \poly$.
\end{lemma}
Lemma~\ref{lemma:unbounded-ed-implies-unbounded-necklaces} and Lemma~\ref{lemma:unbounded-necklaces-implies-geen-kernel} together directly imply Theorem~\ref{thm:main-lower-bound}.

\begin{proof}[Proof of Theorem~\ref{thm:main-lower-bound}]
    By Lemma~\ref{lemma:unbounded-ed-implies-unbounded-necklaces}, graph family $\mathcal G$ contains arbitrarily long uniform $\mathcal F$-necklaces.
    We can therefore apply Lemma~\ref{lemma:unbounded-necklaces-implies-geen-kernel} to conclude that $\mathcal F$-\prob{Minor Free Deletion} does not admit a polynomial kernel in the size of a $\mathcal G$-modulator, assuming $\NP \not \subseteq \coNP / \poly$.
\end{proof}

\subsection{Characterizing graph families with unbounded elimination distance to an $\mathcal F$-minor free graph}

We will prove Lemma~\ref{lemma:unbounded-ed-implies-unbounded-necklaces} here.
Our proof follows the proof by Bougeret et al.\ when they characterize graph families with unbounded bridge-depth~\cite{bridgedepth}.
Similar to their work, we define $\nm_{\mathcal F}(G)$ to be the length of the longest $\mathcal F$-necklace that a graph $G$ contains as a minor for a family of connected graphs $\mathcal F$. Our goal is now to prove the existence of a small set $X$ such that $\nm_{\mathcal F}(G - X) < \nm_{\mathcal F}(G)$ as described in Lemma~\ref{lemma:bounding-function}.

\begin{lemma}\label{lemma:bounding-function}
    Let $\mathcal F$ be a collection of connected planar graphs.
    Then there exists a polynomial function $f_{\mathcal F} \colon \naturals \to \naturals$ such that for any connected graph $G$ with $\nm_{\mathcal F}(G) = t$, there exists a set $X \subseteq V(G)$ with $\abs X \le f_{\mathcal F}(t)$ such that $\nm_{\mathcal F}(G - X) < t$.
\end{lemma}

Bougeret et al.\ showed that one can derive a bounding function when the considered structures satisfy the Erd\H{o}s-P\'osa property~\cite{bridgedepth}.
This also is the case for $\mathcal F$-necklaces when the graphs in $\mathcal F$ are connected and planar, so this approach would be suitable for our purposes as well.
To derive a polynomial bound on the size of $X$, we use a different argument that uses treewidth and grid minors.
We start with the following property of planar graphs.

\begin{proposition}\label{prop:planar-minor-of-grid}
    Any planar graph $G$ on $n$ vertices is a minor of the $4n \times 4n$ grid.
\end{proposition}
\begin{proof}
    We will first modify $G$ by splitting vertices of large degree.
    We do this to ensure that each vertex in the modified graph has degree at most 4.
    For each vertex $u \in V(G)$ with degree $d(u) > 4$, we replace $u$ by a path of $\lfloor d(u) / 2 \rfloor$ vertices.
    Then we connect each neighbor of $u$ to a vertex on this path, such that the maximum degree of each vertex on this path is 4.
    Observe that this is possible since both endpoints of the path can be connected to three neighbors, while the other vertices can be connected to two neighbors.
    Furthermore, notice that we can execute this procedure while ensuring that the modified graph is planar as well.
    We let $G'$ denote the modified graph where each vertex has degree at most 4; observe that $G$ is a minor of this graph.

    As a planar graph on $n$ vertices has at most $3n$ edges, the total degree of $G$ is bounded by $6n$.
    We can use this to bound the number of vertices in $G'$, since
    \[
        \abs{V(G')} \le \sum_{\substack{u \in V(G):\\ d(u) \le 4}} 1 + \sum_{\substack{u \in V(G):\\ d(u) > 4}} \lfloor d(u) / 2 \rfloor \le n + \sum_{u \in V(G)} d(u) / 2 \le 4n.
    \]
    Any planar graph on $n$ vertices that each have degree at most 4 can be drawn on an $n \times n$ grid such that the edges are non-intersecting grid paths~\cite{ntimesngrid}.
    It follows that there exists such an embedding on a grid with side length $4n$ for $G'$.
    Notice that $G'$ is also a minor of this grid, so we conclude that also $G$ is a minor.
\end{proof}

Together with the Excluded Grid Theorem, this proposition leads to the following treewidth bound.

\begin{lemma}\label{lemma:tw-bound-necklaces}
    Let $\mathcal F$ be a collection of connected planar graphs of at most $n$ vertices each.
    There exists a polynomial $f \colon \naturals \to \naturals$ with $f(g) = \bigO(g^{19} \textnormal{poly} \log g)$ such that for any graph $G$ with $\nm_{\mathcal F}(G) = t$, it holds that $\tw(G) < f(4n(t+1))$.
\end{lemma}
\begin{proof}
    Let $f$ be the function guaranteed by the Excluded Grid Theorem such that any graph with treewidth at least $f(g)$ contains the ($g \times g$)-grid as a minor~\cite{improvedgridtheorem}.
    Observe that $f(g) = \bigO(g^{19} \textnormal{poly} \log g)$.
    Now suppose towards a contradiction that the treewidth of $G$ is at least $f(4n(t+1))$.
    Then by the Excluded Grid Theorem, $G$ contains the grid with side length $4(n(t+1))$ as a minor.
    Now observe that an $\mathcal F$-necklace of length $t+1$ has at most $n(t+1)$ vertices when each graph in $\mathcal F$ has at most $n$ vertices and observe that, by Corollary~\ref{corr:planar}, $\mathcal F$-necklaces are planar when the graphs in $\mathcal F$ are planar.
    By Proposition~\ref{prop:planar-minor-of-grid}, such an $\mathcal F$-necklace is therefore a minor of the grid with side length $4(n(t+1))$.
    As $G$ contains this grid as a minor, $G$ also contains an $\mathcal F$-necklace of length $t+1$ as a minor, contradicting that $\nm_{\mathcal F}(G) = t$.
    We conclude that $\tw(G) < f(4n(t+1))$.
\end{proof}

To use this treewidth bound, we need a property similar to~\cite[Lemma 4.6]{bridgedepth}.

\begin{figure}
    \begin{subfigure}{0.45\textwidth}
        \center
        \includegraphics[width=0.75\textwidth]{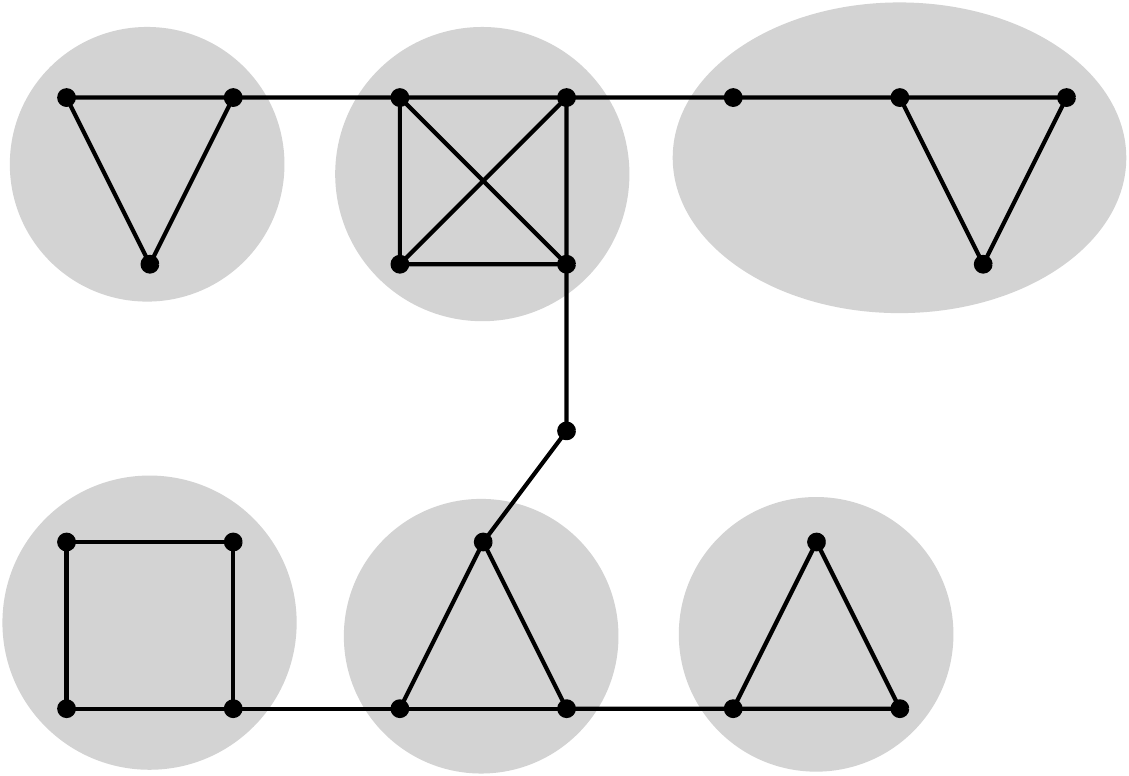}
        \caption{Two minor models of necklaces of length three.}\label{fig:neckl1}
    \end{subfigure}
    \hfill
    \begin{subfigure}{0.45\textwidth}
        \center
        \includegraphics[width=0.75\textwidth]{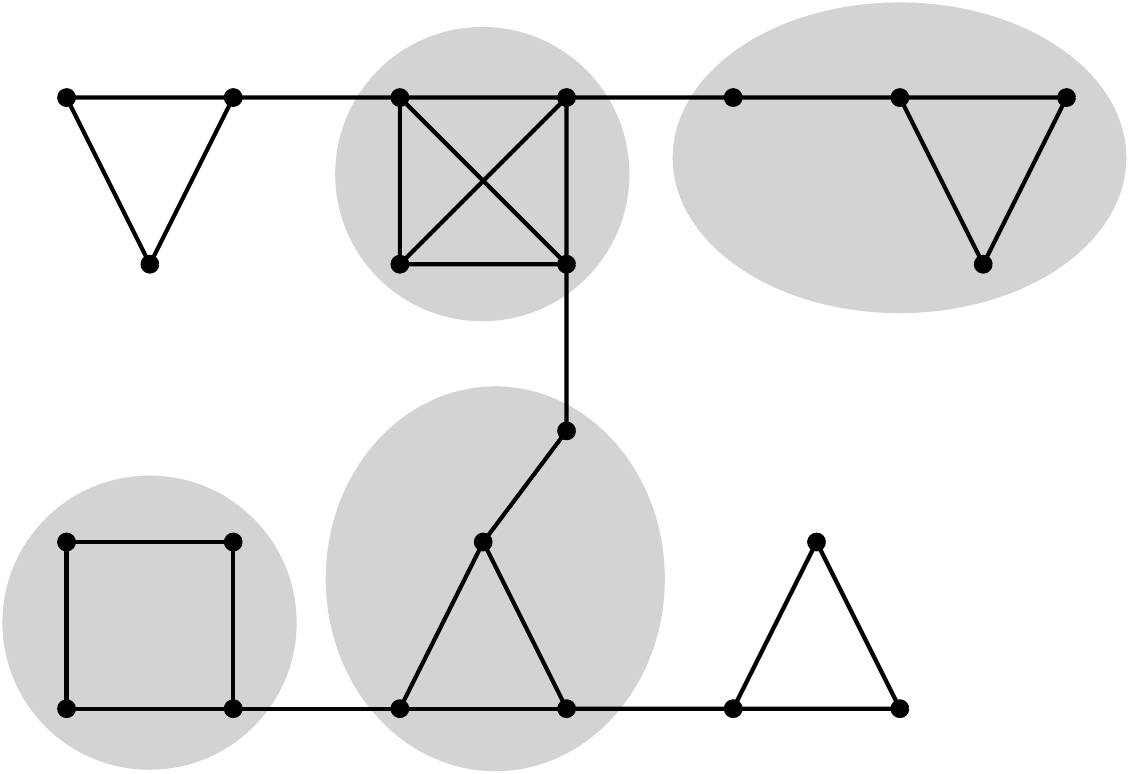}
        \caption{A minor model of a necklace of length four.}
        \label{fig:neck2}
    \end{subfigure}
    \caption{Figure~\ref{fig:neckl1} displays a connected graph where two disjoint minor models of a $\set{K_3}$-necklace of length 3 are highlighted. Figure~\ref{fig:neck2} indicates how these necklaces can be combined into a necklace of length 4.}\label{fig:mergenecklaces}
\end{figure}

\begin{proposition}\label{prop:longest-necklace-unique}
    For any family of connected graphs $\mathcal F$ and connected graph $G$ with $\nm_{\mathcal F}(G) > 0$, any pair of minor models of $\mathcal F$-necklaces of length $\nm_{\mathcal F}(G)$ in $G$ must intersect.
\end{proposition}
\begin{proof}
    Figure~\ref{fig:mergenecklaces} illustrates the main idea of the proof.
    Define $t \coloneqq \nm_{\mathcal F}(G)$.
    Assume towards a contradiction that two disjoint minor models of an $\mathcal F$-necklace of length $t$ exist.
    Then there exist pairwise disjoint sets $S_1, \dots, S_t$, $T_1, \dots, T_t$, such that:
    \begin{itemize}
        \item For each $i \in [t]$, $G[S_i]$ and $G[T_i]$ are connected and they both contain an $\mathcal F$-minor.
        \item For each $i \in [t-1]$, there exists an edge between $S_i$ and $S_{i+1}$, as well as an edge between $T_i$ and $T_{i+1}$.
    \end{itemize}
    Figure~\ref{fig:mergenecklaces} illustrates this setting and gives some intuition on how two of those necklaces can be combined into a longer one.
    Define $S \coloneqq \bigcup_{i \in [t]} S_i$, $\mathcal S = \set{S_1, \dots, S_t}$ and $T$ and $\mathcal T$ analogously.
    Since $G$ is connected, there exists a path between $S$ and $T$.
    Consider such a minimal path $P$, so in particular only its endpoints are in $S \cup T$.
    Let $S_i$ be the set in $\mathcal S$ containing the endpoint of $P$ in $S$ and let $T_j$ be the set in $\mathcal T$ containing the endpoint of $P$ in $T$.
    We will now construct an $\mathcal F$-necklace of length at least $t+1$, contradicting that $\nm_{\mathcal F}(G) = t$.

    If $i \ge \frac{t+1}{2}$, then start the new $\mathcal F$-necklace minor model with sets $S_1, \dots, S_{i-1}$.
    Otherwise, start with $S_t, \dots, S_{i+1}$.
    Observe that also in the latter case, we select at least $t - (i + 1) + 1 \ge t - \frac{t+1}{2} = \frac{t-1}{2}$ sets.
    Now modify set $S_i$ into a set $S_i'$ by including all vertices in $P$ except the vertex in $T_j$.
    We now include $S_i'$ in the necklace.
    Then if $j \ge \frac{t+1}{2}$, we conclude the $\mathcal F$-necklace with sets $T_j, \dots, T_1$; otherwise with $T_j, \dots, T_t$.
    Observe that these are at least $\frac{t+1}{2}$ additional sets.
    We will now argue that these sets satisfy the properties of a necklace.
    \begin{itemize}
        \item The sets are pairwise disjoint, since the sets in $\mathcal S \cup \mathcal T$ are pairwise disjoint and since only one set was modified by adding vertices not occurring in $\mathcal S \cup \mathcal T$.
        \item For each vertex set, the subgraph of $G$ induced by these vertices contains an $\mathcal F$-minor and is connected, since each set is either in $\mathcal S \cup \mathcal T$, or was in $\mathcal S \cup \mathcal T$ before adding vertices on a path with an endpoint in the set.
        \item Between any pair of consecutive sets, there exists an edge:
        the only non-trivial pair of consecutive sets is formed by $S_i'$ and set $T_j$.
        By construction, there exists an edge in $P$ which has one endpoint in $S_i'$ and one in $T_j$, concluding this case as well.
    \end{itemize}
    We conclude that $G$ contains an $\mathcal F$-necklace of length at least $\frac{t-1}{2} + 1 + \frac{t+1}{2} = t+1$ as a minor, contradicting that $\nm_{\mathcal F}(G) = t$.
    Therefore $G$ does not contain two disjoint minor models of an $\mathcal F$-necklace of length $\nm_{\mathcal F}(G)$.
\end{proof}

Proposition~\ref{prop:longest-necklace-unique} is a generalization of the idea that in any connected graph, two paths of maximum length must intersect \bmp{at a vertex}.
Given a graph $G$ with a tree decomposition, we can use this property to identify a vertex in the tree decomposition such that the removal of all vertices in its bag decreases $\nm_{\mathcal F}(G)$.
This result is described in Lemma~\ref{lemma:hitting-necklaces-with-tw-bound}.

\begin{lemma}
    \label{lemma:hitting-necklaces-with-tw-bound}
    Let $\mathcal F$ be a collection of connected graphs.
    Let $G$ be a connected graph with $\tw(G) = w$ and $\nm_{\mathcal F}(G) = t$.
    Then there exists a set $Z \subseteq V(G)$ with $\abs Z \le w+1$ such that $\nm_{\mathcal F}(G-Z) < t$.
\end{lemma}
\begin{proof}
    Take a tree decomposition $\left(X, (B_u)_{u \in V(X)}\right)$ of $G$ of width $w$.
    Now consider a node $t$ of the decomposition such that $G_t$ contains an $\mathcal F$-necklace of length $t$ as a minor, but for none of the children $c$ of $t$ the graph $G_{c}$ contains such a minor.
    Observe that such a node exists, because for the root node $r$, the graph $G_r$ contains such a minor.
    We now claim that $B_t$ hits any minor model of an $\mathcal F$-necklace of length $t$ in $G$.
    First, notice that by definition of a tree decomposition, there cannot be an edge between a vertex $u$ in $G_t - B_t$ and a vertex $v$ in $G - V(G_t)$: if such an edge would exist, then the subtree of the tree decomposition of all bags containing $u$ would intersect the subtree of all bags containing $v$, implying that $B_t$ contains at least one of the vertices.
    It follows that any minor model of an $\mathcal F$-necklace of length $t$ in $G - B_t$ must either be contained in $G_t - B_t$ or $G - V(G_t)$.
    The first set cannot contain such a minor model by construction of node $t$.
    The second set cannot either, as this implies the existence of two disjoint minor models of an $\mathcal F$-necklace of length $t$ in $G$, while Proposition~\ref{prop:longest-necklace-unique} proves that any two of these models must intersect.
    We conclude that $B_t$ is a set of at most $w+1$ vertices hitting all minor models of $\mathcal F$-necklaces of length $t$ in $G$.
\end{proof}

The proof of Lemma~\ref{lemma:bounding-function} follows directly by combining Lemma~\ref{lemma:tw-bound-necklaces} and Lemma~\ref{lemma:hitting-necklaces-with-tw-bound}.

\begin{proof}[Proof of Lemma~\ref{lemma:bounding-function}]
    By Lemma~\ref{lemma:tw-bound-necklaces}, we know that there exists a polynomial $f$ such that $\tw(G) < f(4n(t+1))$.
    Then by Lemma~\ref{lemma:hitting-necklaces-with-tw-bound}, there exists a set $X$ with $\abs X \le f(4n(t+1))$ where $n$ only depends on $\mathcal F$, as it is the maximum number of vertices in a graph in $\mathcal F$.
\end{proof}

We are almost in position to prove Lemma~\ref{lemma:unbounded-ed-implies-unbounded-necklaces} now.
We will first prove the following reformulation and then argue how it implies the lemma.

\begin{lemma}[{Cf.~\cite[Theorem 4.8]{bridgedepth}}]
    \label{lemma:ed-bounded-by-g-of-nm}
    For any finite collection $\mathcal F$ of connected planar graphs, there exists a polynomial function $g$ such that for any graph $G$,
    \[
        \ed_{\neg \mathcal F}(G) \le g(\nm_{\mathcal F}(G)).
    \]
\end{lemma}
\begin{proof}
    Let $\mathcal F$ be a finite collection of connected planar graphs.
    We will prove the statement by induction on $\nm_{\mathcal F}(G) + \abs{V(G)}$ for the function $g$ defined by
    \[
        g(t) = \sum_{i=1}^t f_{\mathcal F}(t),
    \]
    where $f_{\mathcal F}$ denotes the function guaranteed by Lemma~\ref{lemma:bounding-function}.

    If $\nm_{\mathcal F}(G) = 0$, then $G$ does not contain any $\mathcal F$-minor, so $\ed_{\neg \mathcal F}(G) = 0 = g(0)$.

    Now suppose that $\nm_{\mathcal F}(G) = t > 0$ and consider the case where $G$ is connected.
    Then by Lemma~\ref{lemma:bounding-function}, there exists a set $X \subseteq V(G)$ with $\abs{X} \le f_{\mathcal F}(t)$ such that $\nm_{\mathcal F}(G-X) < t$.
    Then for each component $C \in \cc(G-X)$, we know that $C$ is a connected graph with $\nm_{\mathcal F}(C) < t$ and $|V(C)| \le |V(G)|$, so by induction $\ed_{\neg \mathcal F}(C) \le g(\nm_{\mathcal F}(C)) \le g(\nm_{\mathcal F}(G-X)) \le g(t-1)$.
    Since this holds for each component of $G-X$, we obtain by definition that $\ed_{\neg \mathcal F}(G-X) = \max_{C \in \cc(G-X)} \ed_{\neg F}(C) \le g(t-1)$.
    Now observe that $\ed_{\neg \mathcal F}(G) \le \abs X + \ed_{\neg \mathcal F}(G-X)$.
    This follows directly from an inductive proof based on the definition of elimination distance, which implies that there always exists a vertex whose removal decreases the elimination distance.
    This implies that
    \[
        \ed_{\neg \mathcal F}(G) \le \abs X + \ed_{\neg \mathcal F}(G-X) \le f_{\mathcal F}(t) + g(t-1) = g(t).
    \]
    The case where $G$ is disconnected is analogous: for each component $C \in \cc(G)$ we have $\nm_{\mathcal F}(C) \le \nm_{\mathcal F}(G)$ while $\abs{V(C)} < \abs{V(G)}$, so by induction we can bound the elimination distance to an $\mathcal F$-minor free graph of $C$.
    As the elimination distance to an $\mathcal F$-minor free graph of $G$ is the maximum of the elimination distance of its connected components, we thereby bound $\ed_{\neg \mathcal F}(G)$ as well.
\end{proof}

We can now prove Lemma~\ref{lemma:unbounded-ed-implies-unbounded-necklaces}.
Recall that for a finite collection $\mathcal F$ of connected planar graphs, Lemma~\ref{lemma:unbounded-ed-implies-unbounded-necklaces} states that any minor-closed family of graphs $\mathcal G$ with unbounded elimination distance to an $\mathcal F$-minor free graph contains arbitrarily long uniform $\mathcal F$-necklaces.
The proof of this lemma is fairly straightforward now.

\begin{proof}[Proof of Lemma~\ref{lemma:unbounded-ed-implies-unbounded-necklaces}]
    Let $\mathcal G$ be a minor-closed graph family containing graphs with arbitrarily large elimination distance to an $\mathcal F$-minor free graph.
    Suppose towards a contradiction that $\mathcal G$ contains no $\mathcal F$-necklace of length at least $t$ for some $t$.
    Then no graph in $\mathcal G$ contains an $\mathcal F$-necklace of length $t$ as a minor, thereby bounding the elimination distance to an $\mathcal F$-minor free graph of any graph in $\mathcal G$ by $g(t)$ by Lemma~\ref{lemma:ed-bounded-by-g-of-nm}.
    We reach a contradiction and obtain that $\mathcal G$ contains arbitrarily long $\mathcal F$-necklaces, so by Proposition~\ref{prop:always-uniform-necklaces} we conclude that $\mathcal G$ also contains arbitrarily long uniform $\mathcal F$-necklaces.
\end{proof}

\subsection{Lower bound reduction}
Our reduction is based on the following theorem on $d$-\textsc{CNF Satisfiability}, where each clause has at most $d$ literals~\cite{dellmelkebeek2010}.
This theorem uses the notion of an \emph{oracle communication protocol}.
In such a protocol, there is a player with an input for a decision problem who has to answer the decision problem in time polynomial in the input size.
Besides the player, there is an oracle with unbounded resources, but the oracle does not know the input at the start of the protocol.
The goal is to develop a protocol for communication between the player and oracle, such that the player can solve any instance in polynomial time.
A trivial protocol sends the entire instance to the oracle and then sends the solution back, but we are interested in protocols that require less communication.
Theorem~\ref{thm:oracle-cnf} provides a lower bound on the number of bits of communication for such a protocol for $d$-\textsc{CNF}.

\begin{theorem}[{Cf.~\cite[Theorem 1]{dellmelkebeek2010}}]
    \label{thm:oracle-cnf}
    Let $d \ge 3$ be an integer and $\epsilon > 0$.
    Then there is no oracle communication protocol taking $\bigO(n^{d-\epsilon})$ bits of communication that can decide for any $d$-CNF formula $F$ whether $F$ is satisfiable, unless $\NP \subseteq \coNP / \poly$.
\end{theorem}

Observe that a protocol sending the entire instance to the oracle takes $\bigO(n^d)$ bits of communication.
The theorem implies the following theorem on \textsc{CNF Satisfiability}, where clauses can have arbitrarily many variables.
This theorem is also described in~\cite{jansen2014}.
\begin{theorem}\label{thm:cnfsat}
There is no polynomial time algorithm that, given a CNF formula $F$ on $n$ variables, constructs an instance $I$ of a fixed decision problem $Q$ such that
\begin{enumerate}
    \item $F$ is satisfiable if and only if $I$ is a \textsc{Yes}-instance of $Q$,
    \item $\abs I = \bigO(n^c)$ for some constant $c$,
\end{enumerate}
unless $\NP \subseteq \coNP / \poly$.
\end{theorem}
\begin{proof}
    Assume towards a contradiction that such a polynomial-time kernelization algorithm would exist and that the kernel has size $\bigO(n^c)$ for some constant $c$.
    Without loss of generality, take $c \ge 2$.
    We can then construct an oracle communication protocol contradicting Theorem~\ref{thm:oracle-cnf}.
    Consider $(c+1)$-CNF, where each clause has at most $c+1$ variables.
    Using this kernelization algorithm, the player can first reduce the instance to an instance of size $\bigO(n^c)$.
    The protocol is concluded by communicating this instance to the oracle and letting the oracle solve the problem.
    Observe that this protocol took only $\bigO(n^c)$ bits of communication.
    This directly contradicts Theorem~\ref{thm:oracle-cnf} which stated the non-existence of any $\bigO(n^{c+1-\epsilon})$ communication protocol for any $\epsilon>0$ assuming $\NP \not \subseteq \coNP / \poly$.
    Therefore, such a kernelization algorithm does not exist under these complexity assumptions.
\end{proof}

Theorem~\ref{thm:cnfsat} implies that, under common hardness assumptions, there does not exist a kernelization algorithm for the \prob{CNF Satisfiability} problem that outputs a polynomial kernel in the number of variables.
We can use this theorem to prove the non-existence of a kernelization algorithm for $\mathcal F$-\prob{Minor Free Deletion} with as parameter the size of a modulator to a graph with constant elimination distance to an $\mathcal F$-minor free graph.
Our reduction will use the following graphs.
An example is illustrated in Figure~\ref{fig:jhv}.

\begin{definition}[$J_{H, v, x, y}$]
    \label{def:JHv}
    Let $H$ be a graph, $v, x, y \in V(H)$ with $x \neq y$ and let $c = \abs{V(H)}$.
    Define $J_{H, v, x, y}$ to be the graph which can be constructed with the following two steps.
    \begin{itemize}
        \item We create $c-1$ vertices $x_1, \dots, x_{c-1}$ for which we define the set $X \coloneqq \set{x_1, \dots, x_{c-1}}$.
        We then add precisely those edges such that $J_{H, v, x, y}[X]$ is isomorphic to $H - \set v$.
        We repeat this construction for a set of vertices $Y \coloneqq \set{y_1, \dots, y_{c-1}}$ such that $J_{H, v, x, y}[Y]$ is isomorphic to $H-\set v$.
        \item Then for any pair of $x_i \in X$ and $y_j \in Y$, we add $c-2$ extra vertices to the graph and let $D_{i, j}$ be the set containing these vertices.
        We then add edges to the graph such that $J_{H, v, x, y}[D_{i, j} \cup \set{x_i, y_j}]$ is isomorphic to $H$.
        In particular, there exists an isomorphism where $x_i$ is the image of $x$ and $y_j$ is the image of $y$.
    \end{itemize}
    In total, the graph will have $2(c-1) + (c-1)^2(c-2)$ vertices.
\end{definition}
\begin{figure}
    \centering
    \includegraphics[width=0.35 \textwidth]{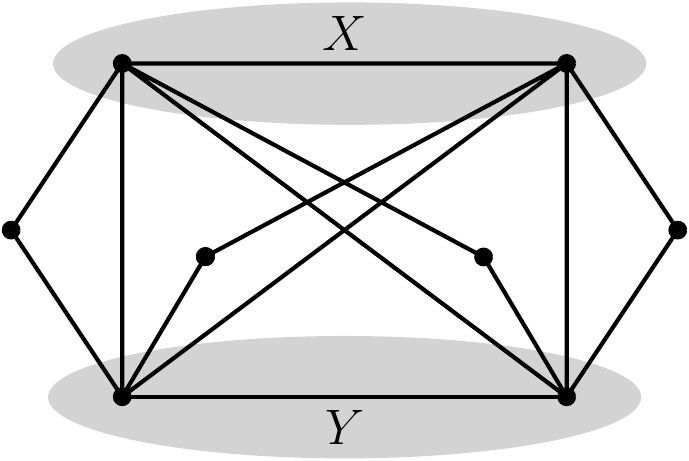}
    \caption{Illustration of $J_{H, v, x, y}$ with $H = {K_3}$ and $v, x, y \in V(H)$.
    Observe that the subgraph induced by $X$ is indeed isomorphic to $K_3 - \set{v}$.}\label{fig:jhv}
\end{figure}

The choice of vertices $v, x$ and $y$ will not be important, but they ensure that the graph is uniquely defined.
The constructed graph has some useful properties.

\begin{lemma}
    \label{lemma:JHv-no-minor-then-X-or-Y}
    Let $H$ be a biconnected graph on at least three vertices and let $v, x, y \in V(H)$ with $x \neq y$.
    Let $J_{H, v, x, y}$ be the graph defined in Definition~\ref{def:JHv} with $X$, $Y$ and $c$ defined accordingly.
    Let $Z \subseteq V(J_{H, v, x, y})$ be a vertex set with $\abs Z \le c-1$.
    If $J_{H,v, x, y} - Z$ contains no $H$-minor, then $Z$ is equal to $X$ or $Y$.
\end{lemma}
\begin{proof}
    Assume towards a contradiction that $Z$ is not equal to $X$ or $Y$.
    We distinguish a few cases.
    \begin{itemize}
        \item If $Z$ does not intersect $X \cup Y$, then $Z$ cannot hit all $H$-minors which we added for each pair of vertices $x \in X$ and $y \in Y$:
        there are $(c-1)^2$ such pairs, giving a set of $(c-1)^2$ $H$-minors that only intersect each other in $X \cup Y$.
        As $(c-1)^2 > c-1$ for $c \ge 3$, the graph $J_{H,v, x, y} - Z$ contains an $H$-minor in this case.
        \item Suppose $Z$ is included in $X \cup Y$.
        Since $Z$ is not equal to $X$ and not equal to $Y$, there exist vertices $x_i \in X$ and $y_j \in Y$ such that $x_i, y_j \not \in Z$.
        As $Z$ contains no vertices outside $X \cup Y$, it cannot hit the $H$-minor that was constructed for this pair, so $J_{H,v, x, y} - Z$ contains an $H$-minor here as well.
        \item Now assume neither of those cases apply.
        Let $k = \abs{Z \cap X}$ and $\ell = \abs{Z \cap Y}$ and observe that $k + \ell \le c-2$.
        Then the number of vertices in $X \setminus Z$ equals $c - 1 - k$ and the size of $Y \setminus Z$ equals $c - 1 - \ell$.
        It follows that the number of pairs of a vertex in $X \setminus Z$ and one in $Y \setminus Z$ equals
        \begin{align*}
        (c - 1 - k)(c - 1 - \ell) &= (c-1)^2 - (c-1)(k + \ell) + k \ell \\
        & \ge (c-1)^2-(c-1)(c-2) \\
        &= c-1.
        \end{align*}
        We now obtained at least $c-1$ $H$-minors where any pair of them can only intersect in a vertex in $(X \cup Y) \setminus Z$.
        Therefore, if these minors are all hit by $Z$, then $Z$ needs to contain at least $c-1$ vertices outside $X \cup Y$.
        However, at least one vertex of $Z$ is in $X \cup Y$, contradicting that $\abs Z \le c-1$.
    \end{itemize}
    We conclude that if $Z$ is not $X$ or $Y$, then $J_{H,v, x, y}$ contains an $H$-minor, thereby concluding the proof.
\end{proof}

The implication in Lemma~\ref{lemma:JHv-no-minor-then-X-or-Y} is actually an equivalence.
For the other direction, we will prove a slightly stronger statement.

\begin{lemma}
    \label{lemma:JHv-X-or-Y-then-no-minor}
    Let $H$ be a biconnected graph on at least three vertices and let $v, x, y \in V(H)$ with $x \neq y$.
    Let $J_{H, v, x, y}$ be the graph defined in Definition~\ref{def:JHv} with $X$, $Y$ and $c$ defined accordingly.
    If $Z \subseteq V(J_{H, v, x, y})$ is equal to $X$ or $Y$, then any biconnected graph that $J_{H,v, x, y} - Z$ contains as a minor is a proper minor of $H$.
\end{lemma}
\begin{proof}
    Assume without loss of generality that $Z = Y$.
    Now assume towards a contradiction that $J_{H,v, x, y} - Y$ contains an $H'$-minor for some biconnected $H'$ that is not a proper minor of $H$
    Observe that by Proposition~\ref{prop:biconn}, a biconnected component of $J_{H,v, x, y} - Y$ must contain an $H'$-minor.
    However, the only biconnected components of $J_{H,v, x, y} - Y$ are the induced subgraphs by the vertices $D_{i,j} \cup \set{x_i}$ for $i \in [c-1]$.
    These subgraphs are minors of $H$, so they cannot contain an $H'$-minor when $H'$ is a proper minor of $H$. 
\end{proof}

We can now give the reduction.

\begin{lemma}
    \label{lemma:fvsreduction}
    Let $\mathcal F$ be a finite collection of biconnected graphs on at least three vertices.
    Let $H \in \mathcal F$, let $x, y \in V(H)$ and consider the uniform necklace structure $(H, x, y)$.
    There exists a polynomial time algorithm that, given a CNF formula on $n$ variables $x_1, \dots, x_n$, outputs a graph $G$ together with a set $X \subseteq V(G)$ of size $k = \bigO_{\mathcal F}(n)$ and an integer $t$, such that
    \begin{enumerate}
        \item $F$ is satisfiable if and only if $G$ has a solution to $\mathcal F$-\prob{Minor Free Deletion} of size at most $t$,
        \item Each connected component of $G - X$ is a minor of a uniform necklace with structure $(H, x, y)$.
    \end{enumerate}
\end{lemma}
\begin{proof}
    We will also assume that for any $H \in \mathcal F$, no proper minor of $H$ is contained in $\mathcal F$.
    Otherwise, we can remove $H$ from $\mathcal F$ in the context of $\mathcal F$-\prob{Minor Free Deletion}, as the deletion of all $H$-minors follows from the deletion of the proper minors of $H$.

    We will start our construction of $G$ by picking vertices $v, x, y \in V(H)$ with $x \neq y$ and let $c = \abs{V(H)} = \bigO_{\mathcal F}(1)$.
    Given a CNF formula $F$ on variables $x_1, \dots, x_n$, the algorithm executes the following three steps which are illustrated by Figure~\ref{fig:reduction}:

    \begin{enumerate}
        \item For each variable $x_i$, we add a variable gadget $J_i$ to $G$ which is a copy of the graph $J_{H,v, x, y}$.
        Instead of using $X$ and $Y$ for the two subsets of $V(J_i)$ defined earlier, we use the names $V_{x_i} = \set{v_{x_i, 1}, \dots, v_{x_i, c-1}}$ and $V_{\neg x_i} = \set{v_{\neg x_i, 1}, \dots, v_{\neg x_i, c}}$.
        Notice the direct correspondence between literals and these sets.
        \item For each clause in $F$, we are going to add a necklace to $G$.
        For an integer $t$, we let $N_t$ denote the uniform necklace with structure $(H, x, y)$ of length $t$.
        For each clause in $F$ with length $m$, we add a copy of $N_{m}$ to the graph.
        Each literal in this clause corresponds to a bead in the necklace.
        For $i \in [m]$, we let $\ell_i$ be the $i$th literal in this clause.
        We are now going to connect the $i$th bead to the vertices in $V_{\ell_i}$ for each $i \in [m]$.
        To do this, we pick a vertex $u$ in the $i$th bead that is not adjacent to the other beads.
        As $H$ has at least three vertices, such a vertex exists.
        We then connect $u$ to $V_{\ell_i}$, such that the subgraph induced by $V_{\ell_i} \cup \set u$ is isomorphic to $H$.
        By definition of $J_{H,v, x, y}$, we can ensure this by only adding edges between $u$ and $V_{\ell_i}$.
        \item It remains to add one extra copy of $H$ to $G$, whose vertex set we denote with $S_0$.
        We remove one edge from this copy of $H$, say with endpoints $a$ and $b$.
        Then for every necklace that we added, we add an edge between $a$ and the first bead, and an edge between $b$ and the last bead of the necklace.
        We make sure that the endpoints in the necklace are not connected to a variable gadget.
    \end{enumerate}

    Observe that, for fixed $\mathcal F$, the entire procedure can be executed in time polynomial in the size of $F$.

    \begin{figure}
        \centering
        \includegraphics[width=0.9 \textwidth]{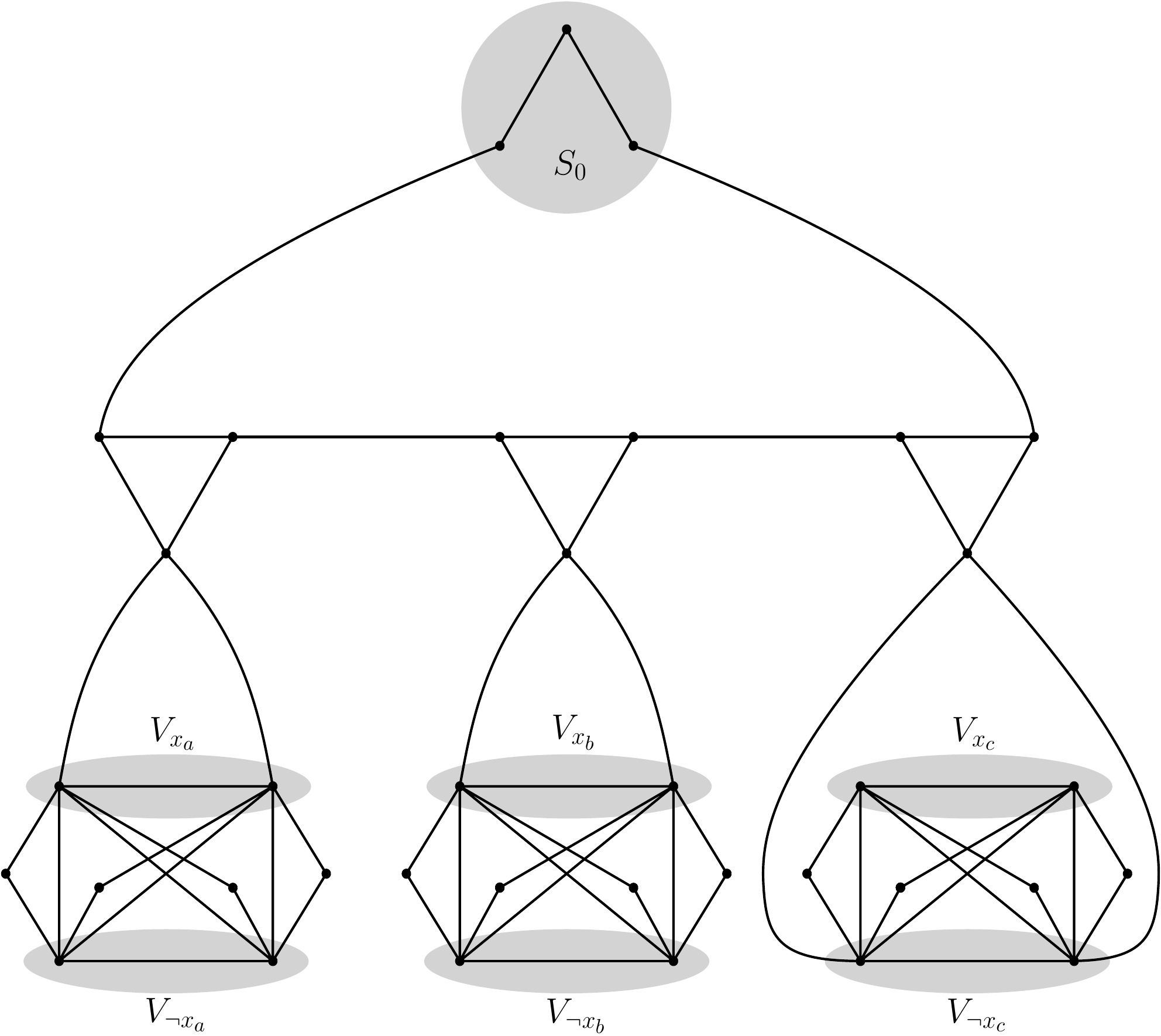}
        \caption{Illustration of the reduction in the proof of Lemma~\ref{lemma:fvsreduction} for a clause $(x_a \vee x_b \vee \neg x_c)$.
        Let $\ell_i$ denote the $i$th literal of this clause.
        We pick $H = K_3$ and we pick distinct $x, y \in V(H)$, so the uniform necklace structure becomes $(K_3, x, y)$.
        We add such a uniform necklace of length $3$.
        The first and last bead are connected to $S_0$.
        We then connect each bead to a variable gadget.
        For the $i$th bead, we pick a vertex $u$ in this bead which is not adjacent to the next or previous bead, and neither to $S_0$.
        We then ensure that the subgraph induced by $u$ and $V_{\ell_i}$ equals $K_3$.
        Observe that, since the last literal of the clause is a negated variable, the last bead is connected to $V_{\neg x_c}$.
        }\label{fig:reduction}
    \end{figure}

    Let $X \coloneqq \left( \bigcup_{i \in [n]} (V_{x_i} \cup V_{\neg x_i}) \right) \cup S_0$.
    Observe that $\abs X = 2n(c-1) + c = \bigO_{\mathcal F}(n)$.
    We are now interested in the structure of the connected components of $G - X$.
    For each clause in $F$, we added a uniform necklace with structure $(H, x, y)$ in the second step.
    Since the neighbors of this necklace in $G$ are all in $X$, these necklaces are connected components of $G-X$.
    The other connected components of $G-X$ are contained in the variable gadgets which were added in the first step.
    Each of these is a copy of $H$ where two vertices are removed, so these components are minors of $H$.
    Therefore, each connected component of $G-X$ is a minor of a uniform necklace with structure $(H, x, y)$.

    Let $w$ be the total number of occurrences of literals in $F$, i.e.
    \[
        w \coloneqq \sum_{\text{clause $C$ in $F$}} \text{\#literals in $C$}
    \]
    We claim that $F$ is satisfiable if and only if $G$ has a solution to $\mathcal F$-\prob{Minor Free Deletion} of size at most $w+(c-1)n$.
    We will first prove the `only if' direction.

    \begin{claim}
        If $F$ is satisfiable, then $G$ has a solution to $\mathcal F$-\prob{Minor Free Deletion} of size at most $w+(c-1)n$.
    \end{claim}
    \begin{innerproof}
        Suppose that $F$ has a satisfying truth assignment.
        Based on this truth assignment, we will construct a solution $Z$ for $\mathcal F$-\prob{Minor Free Deletion} of size $w+(c-1)n$ in the corresponding graph $G$.
        For every variable $x_i$ that is set to \textsc{True}, we add the $c-1$ vertices in $V_{x_i}$ to $Z$.
        For every variable set to \textsc{False}, we add $V_{\neg x_i}$ to $Z$.
        These vertices are in the variable gadgets added in the first step.
        Now, $Z$ contains $(c-1)n$ vertices.

        Then for each clause in $F$, we consider each literal $\ell_i$ in this clause and consider the necklace in $G$ corresponding to this clause.
        If $\ell_i$ evaluates to \textsc{True} according to the truth assignment, we add a vertex from the $i$th bead to $Z$ that is adjacent to the previous or next bead.
        Otherwise, we add the vertex in this bead adjacent to $V_{\ell_i}$.
        This procedure adds another $w$ vertices to $Z$, ensuring that $Z$ has size $w+(c-1)n$.

        We will now prove that $G-Z$ contains no $\mathcal F$-minors.
        The graph $G-Z$ contains multiple connected components.
        We will argue that none of these components contains an $\mathcal F$-minor, as it implies directly that $G-Z$ has no $\mathcal F$-minors.

        Consider a variable gadget $J_i$ corresponding to some variable $x_i$, and assume without loss of generality (by symmetry) that $x_i$ evaluates to \textsc{True} in the given truth assignment.
        The only vertices in $J_i$ that were connected to other vertices in $G$ are those in $V_{x_i}$ and those in $V_{\neg x_i}$.
        The first set is removed.
        For the second set, observe that these vertices are only connected to beads in necklaces where the corresponding literal is $\neg x_i$, and that precisely these vertices in those beads are removed as $\neg x_i$ evaluates to \textsc{False}.
        It follows that $J_i - V_{x_i}$ is a connected component in $G-Z$.
        By Lemma~\ref{lemma:JHv-X-or-Y-then-no-minor}, any biconnected minor of $J_i - V_{x_i}$ is a proper minor of $H$.
        By assumption, for each pair of graphs $H_1, H_2 \in \mathcal F$, $H_1$ is not a minor of $H_2$ and $H_2$ is not a minor of $H_1$.
        Since all graphs in $\mathcal F$ are biconnected, it follows that the connected component $J_i - V_{x_i}$ contains no $\mathcal F$-minors.

        After analyzing the connected components containing vertices of variable gadgets, it remains to analyze the other connected components of $G-Z$.
        Define the set $\mathcal F'$ to be the set of all graphs that can be obtained by either removing a vertex or by removing an edge from a graph in $\mathcal F$.
        Observe that $(G-Z)[S_0] \in \mathcal F'$, and similarly the subgraph induced by what remains of a bead of a necklace is in $\mathcal F'$.
        Besides, since at least one literal evaluates to \textsc{True} in each clause, we always remove at least one edge between some pair of consecutive beads in a necklace.
        Hence, each necklace in $G$ is now broken into at least two $\mathcal F'$-necklaces in $G-Z$.
        This also implies that in each necklace there does not exist a path through the necklace from the first to the last bead.
        Therefore, there exists at most one edge between each of those $\mathcal F'$-necklaces and vertex set $S_0$.
        Now consider the connected component that includes $S_0$.
        It follows that each biconnected component of at least three vertices of this component is either included in a bead or in $S_0$.
        In both cases, it is a minor of a graph in $\mathcal F'$.
        Observe that no graph in $\mathcal F'$ contains an $\mathcal F$-minor:
        then this graph would be a proper minor of $H$ since $\mathcal F'$ contains only proper subgraphs of $H$.
        It follows that no biconnected component on at least three vertices contains an $\mathcal F$-minor, so this component contains no $\mathcal F$-minor by Proposition~\ref{prop:biconn}.
        It remains to observe that all other components of $G-Z$ are $\mathcal F'$-necklaces which are not connected to any other vertices.
        By Lemma~\ref{lemma:necklace-minors}, these components contain no $\mathcal F$-minors either.
        Since no component of $G-Z$ contains an $\mathcal F$-minor, $G-Z$ has no $F$-minors.
    \end{innerproof}

    It remains to prove the opposite implication.

    \begin{claim}
        If $G$ has a solution to $\mathcal F$-\prob{Minor Free Deletion} of size at most $w+(c-1)n$, then $F$ is satisfiable.
    \end{claim}
    \begin{innerproof}
        To this end, we will first observe that there are $w+(c-1)n$ pairwise disjoint $\mathcal F$-minors in $G$.
        Each literal contributes such an $H$-minor in the bead in the corresponding $\set H$-necklace, giving $w$ pairwise disjoint $\mathcal F$-minors.
        Now consider a variable gadget $J_i$ corresponding to variable $x_i$.
        It contains the vertex sets $V_{x_i} = \set{v_{x_i, 1}, \dots, v_{x_{c-1}}}$ and $V_{\neg x_i} = \set{v_{\neg x_i, 1}, \dots, v_{\neg x_{c-1}}}$.
        For each pair of a vertex in $V_{x_i}$ and a vertex in $V_{\neg x_i}$, we created an $H$-minor in this gadget.
        We can identify $c-1$ pairwise disjoint $H$-minors in $G$ by considering for each $j \in [c-1]$ the $H$-minor created for the pair of $v_{x_i, j}$ and $v_{\neg x_i, j}$.
        Observe that these minors are pairwise disjoint and contained in $J_i$.
        Therefore each variable gadget contributes $c-1$ $H$-minors, giving a total of $w+(c-1)n$ $H$-minors which are all pairwise disjoint.

        Now consider a variable gadget $J_i$ for variable $x_i$.
        If there exists a solution to $\mathcal F$-\prob{Minor Free Deletion} within this variable gadget of at most $c-1$ vertices, then by Lemma~\ref{lemma:JHv-no-minor-then-X-or-Y} the solution is equal to $V_{x_i}$ or $V_{\neg x_i}$.

        Now consider a solution $Z$ for $\mathcal F$-minor deletion in $G$ of at most $w+(c-1)n$ vertices.
        Since we can identify $w+(c-1)n$ pairwise disjoint $H$-minors, $Z$ contains one vertex from each of the identified $H$-minors in $G$.
        It follows that $Z$ contains $c-1$ vertices from each variable gadget.
        Then for a variable $x_i$, the intersection of $Z$ with $V(J_i)$ equals precisely $V_{x_i}$ or $V_{\neg x_i}$ by Lemma~\ref{lemma:JHv-no-minor-then-X-or-Y}.

        We can now construct a truth assignment for $F$.
        For each variable $x_i$, we verify which vertices $Z$ contains from $J_i$.
        If $V_{x_i}$ is included (and therefore $V_{\neg x_i}$ is disjoint from $Z$), we set $x_i$ to \textsc{True}, and if $V_{\neg x_i}$ is included we set $x_i$ to \textsc{False}.
        It remains to prove that this yields a valid truth assignment.

        To this end, we first consider an $\mathcal F$-necklace corresponding to a clause.
        Observe that $Z$ contains precisely one vertex from each bead, as each bead contains one of the $w+(c-1)n$ pairwise disjoint $H$-minors.
        Suppose towards a contradiction that in each bead, $Z$ contains the vertex which we connected to a variable gadget.
        Since this vertex is not adjacent to other beads by construction and since $H$ is biconnected, it follows that there still exists a path from $a \in S_0$ through this necklace to $b \in S_0$.
        By contracting this path to a single edge, we obtain a copy of the graph $H$ on vertex set $S_0$, contradicting that $Z$ hits all $H$-minors.
        It follows that in each necklace corresponding to a clause, we can identify a bead where the vertex connected to the corresponding variable gadget is not in $Z$.
        Let this vertex be $u$ in a bead corresponding to literal $\ell$.
        We now claim that $\ell$ evaluates to \textsc{True}, implying that such a literal exists in every clause and thereby proving that $F$ is satisfiable.
        Assume towards a contradiction that $\ell$ evaluates to \textsc{False}, then in the corresponding variable gadget, none of the vertices of $V_{\ell}$ were included in $Z$.
        However, recall that we created an $H$-subgraph on vertex set $V_{\ell} \cup \set{ v }$.
        As none of these vertices is in $Z$, this subgraph contradicts $Z$ being a solution to $\mathcal F$-\prob{Minor Free Deletion}.
        It follows that $\ell$ evaluates to \textsc{True}, so every clause in $F$ contains a literal that evaluates to \textsc{True}.
        This concludes the claim that $F$ is satisfiable.
    \end{innerproof}

    These claims prove the correctness of our reduction and conclude the proof of Lemma~\ref{lemma:fvsreduction}.
\end{proof}

We are finally able to prove Lemma~\ref{lemma:unbounded-necklaces-implies-geen-kernel}.
We formalize this lemma with the following theorem and prove the theorem.

\begin{theorem}[Formalization of Lemma~\ref{lemma:unbounded-necklaces-implies-geen-kernel}]
    \label{thm:no-kernelization-of-f-minor-free-deletion}
    Let $\mathcal F$ be a finite collection of biconnected graphs on at least three vertices.
    Let $\mathcal G$ be a minor-closed graph family containing arbitrarily long uniform $\mathcal F$-necklaces.
    There is no polynomial time algorithm that, given a graph $G$, integer $t$ and vertex set $X$ of size $k$ such that $G-X$ is in $\mathcal G$, constructs an instance $I$ of a fixed decision problem $Q$ such that
    \begin{enumerate}
        \item $G$ has a solution to $\mathcal F$-\prob{Minor Free Deletion} of size at most $t$ if and only if $I$ is a \textsc{Yes}-instance of $Q$,
        \item $\abs I = \bigO(k^c)$ for some constant $c$,
    \end{enumerate}
    unless $\NP \subseteq \coNP / \poly$.
\end{theorem}
\begin{proof}
    We will prove this theorem by contradiction: if such an algorithm would exist, we can use it to give a kernel for the \textsc{CNF Satisfiability} problem which would contradict Theorem~\ref{thm:cnfsat}.
    This can be achieved by reducing a provided CNF formula to an instance of $\mathcal F$-\prob{Minor Free Deletion}.

    Let $F$ be a CNF formula on $n$ variables.
    By Proposition~\ref{prop:always-uniform-necklaces}, there exists a uniform necklace structure $(H, x, y)$ with $H \in \mathcal F$ and $x, y \in V(H)$ such that all uniform necklaces with this structure are contained in $\mathcal G$.
    For this uniform necklace structure $(H, x, y)$, let $G$ be the graph obtained by Lemma~\ref{lemma:fvsreduction} with a set $X$ of size $k = \bigO_{\mathcal F}(n)$ such that all connected components of $G-X$ are minors of uniform necklaces with this structure.
    Then observe that $G-X$ itself is also a minor of a sufficiently long uniform necklace with the same structure.
    Let $t$ be the integer such that $F$ is satisfiable if and only if $G$ has a solution to $\mathcal F$-\prob{Minor Free Deletion} of size at most $t$.

    If such an algorithm would exist, we can now construct an instance $I$ of a fixed decision problem $Q$ such that $G$ has a solution to $\mathcal F$-\prob{Minor Free Deletion} of size at most $t$ if and only if $I$ is a \textsc{Yes}-instance of $Q$, while $\abs I = \bigO(k^c)$.
    Since $k = \bigO_{\mathcal F}(n)$, the bound on the size of $I$ becomes $\abs I = \bigO(n^c)$ when $\mathcal F$ is a finite collection of graphs.
    However, since $G$ has a solution to $\mathcal F$-\prob{Minor Free Deletion} of size at most $t$ if and only if $F$ is satisfiable, we obtain a contradiction with Theorem~\ref{thm:cnfsat}.
    Therefore such an algorithm cannot exist, concluding the proof of Theorem~\ref{thm:no-kernelization-of-f-minor-free-deletion}.
\end{proof}

\subsection{Proof of Theorem~\ref{thm:main}}
By combining our previous results, we can now prove Theorem~\ref{thm:main}.

\begin{customthm}{1}
Let $\mathcal G$ be a minor-closed graph family.
Then \prob{Feedback Vertex Set} admits a polynomial kernel in the size of a $\mathcal G$-modulator if and only if $\mathcal G$ has bounded elimination distance to a forest, assuming that $\NP \not \subseteq \coNP / \poly$.
\end{customthm}
\begin{proof}
    Let $\mathcal G_\eta$ be the set of graphs with elimination distance to a forest at most $\eta$.
    Then by Theorem~\ref{thm:main-kernel}, \prob{Feedback Vertex Set} admits a polynomial kernel in the size of a $\mathcal G_{\eta}$-modulator $X$.
    Now let $\mathcal H$ be a family of graphs where each graph has elimination distance to a forest at most $\eta$.
    Then Theorem~\ref{thm:main-kernel} implies directly that \prob{Feedback Vertex Set} also admits a polynomial kernel in an $\mathcal H$-modulator, since such a modulator is always also an $\mathcal G_\eta$-modulator since $\mathcal H \subseteq \mathcal G_\eta$.
   
   Regarding the lower bound, Theorem~\ref{thm:main-lower-bound} states that for any minor-closed graph family $\mathcal G$ and any finite set $\mathcal F$ of biconnected planar graphs on at least three vertices, $\mathcal F$-\prob{Minor Free Deletion} does not admit a polynomial kernel in the size of a $\mathcal G$-modulator, unless $\NP \subseteq \coNP / \poly$.
   By taking $\mathcal F = \set{K_3}$, we obtain the stated lower bound.
\end{proof}

    \section{Conclusion and discussion}\label{sec:discussion}

    We conclude that the elimination distance to a forest characterizes the \prob{Feedback Vertex Set} problem in terms of polynomial kernelization.
For a minor-closed graph family $\mathcal G$, the problem admits a polynomial kernel in the size of a $\mathcal G$-modulator if and only if $\mathcal G$ has bounded elimination distance to a forest, assuming $\NP \not \subseteq \coNP / \poly$.
In particular, this implies that \prob{Feedback Vertex Set} does not admit a polynomial kernel in the deletion distance to a graph of constant bridge-depth under the mentioned hardness assumption.
We also generalize the lower bound to other $\mathcal F$-\prob{Minor Free Deletion} problems where $\mathcal F$ contains only biconnected planar graphs on at least three vertices.
It remains unknown whether such a lower bound also generalizes to collections of graphs $\mathcal F$ that contain non-planar graphs.

An interesting open problem is whether similar polynomial kernels can be obtained for other $\mathcal F$-\prob{Minor Free Deletion} problems.
Regarding the field of fixed parameter tractable algorithms, it was recently shown~\cite{eldistapproxes} \bmp{that for any set~$\mathcal{F}$ of connected graphs,} $\mathcal F$-\prob{Minor Free Deletion} admits an FPT algorithm when parameterized by the elimination distance to an $\mathcal F$-minor free graph (or even $\mathcal H$-treewidth when $\mathcal H$ is the class of $\mathcal F$-minor free graphs).
This generalizes known FPT algorithms for the natural parameterization by solution size.
Regarding polynomial kernels, $\mathcal F$-\prob{Minor Free Deletion} problems admit a polynomial kernel in the solution size when $\mathcal F$ contains a planar graph~\cite{FominLMS12}.
\bmp{Do polynomial kernels exist when the problem is parameterized by a modulator to a graph of constant elimination distance to being $\mathcal{F}$-minor free?}
%
%
%
 \bibliographystyle{splncs04}
 \bibliography{bib}
 
\end{document}